\documentclass{psa}

\usepackage{amsmath}
\usepackage{mathtools}
\usepackage{mathrsfs} 
\usepackage{thmtools} 
\usepackage{thm-restate}
\usepackage{hyperref}
\usepackage{csquotes}

\usepackage{tikz}
\usepackage{xfp} 
\usepackage[outline]{contour} 
\usetikzlibrary{decorations.markings,decorations.pathmorphing}
\usetikzlibrary{decorations.pathreplacing}
\usetikzlibrary{angles,quotes} 
\usetikzlibrary{arrows.meta} 
\usetikzlibrary{patterns, patterns.meta}
\contourlength{1.4pt}

\newcommand{\calI}{\mathscr{I}} 
\tikzset{>=latex} 
\colorlet{myred}{red!80!black}
\colorlet{myblue}{blue!80!black}
\colorlet{mygreen}{green!80!black}
\colorlet{mydarkred}{red!50!black}
\colorlet{mydarkblue}{blue!50!black}
\colorlet{mylightblue}{mydarkblue!6}
\colorlet{mypurple}{blue!40!red!80!black}
\colorlet{mydarkpurple}{blue!40!red!50!black}
\colorlet{mylightpurple}{mydarkpurple!80!red!6}
\colorlet{myorange}{orange!40!yellow!95!black}
\tikzstyle{cone}=[mydarkblue,line width=0.2,top color=blue!60!black!30,
                  bottom color=blue!60!black!50!red!30,shading angle=60,fill opacity=0.9]
\tikzstyle{cone back}=[mydarkblue,line width=0.1,dash pattern=on 1pt off 1pt]
\tikzstyle{world line}=[myblue!60,line width=0.4]
\tikzstyle{world line t}=[mypurple!60,line width=0.4]
\tikzstyle{particle}=[mygreen,line width=0.5]
\tikzstyle{photon}=[-{Latex[length=4,width=3]},myorange,line width=0.4,decorate,
                    decoration={snake,amplitude=0.9,segment length=4,post length=3.8}]
\tikzstyle{photon2}=[-{Latex[length=4,width=3]},myorange,line width=0.4,decorate,
                    decoration={snake,amplitude=0.9,segment length=4,post length=3.8}]
\tikzstyle{singularity}=[myred,line width=0.6,decorate,
                         decoration={zigzag,amplitude=2,segment length=6.17}]

\renewcommand{\If}{\mathscr{I}^+}
\newcommand{\Ip}{\mathscr{I}^-}

\newcommand{\BH}{\mathcal{B}}

 \DeclareSymbolFont{borrowedletters}{OML}{txmi}{m}{it}
 \DeclareMathSymbol{\psi}{\mathord}{borrowedletters}{"20}

\makeatletter
\renewcommand{\printhistory}{{\par\addvspace{10pt}%
\historyfont\noindent%
\ifx\@history\empty\gdef\@history{Received xx xxx xxxx}\fi
\vskip0pt
\normalsize{\noindent Accepted for Publication in \textit{Philosophy of Science}\par \today}\par}
}%
\makeatother

\makeatletter
\let\c@theorem\relax
\makeatother

\theoremstyle{theormstyle}
\newtheorem{theorem}{Theorem}
\newtheorem{lemma}[theorem]{Lemma}

\theoremstyle{definition}
\newtheorem{definition}{Definition}

\begin{document}

\doival{10.1017/xxxxx}
\jnlPage{1}{25}
\jnlDoiYr{2026}
\lefttitle{Is Black Hole Evaporation Prediction Friendly?}

\papertitle{RESEARCH ARTICLE}

\title{Is Black Hole Evaporation Prediction Friendly?}

\author{Dominic Ryder$^{1}$}

\affil{$^{1}$Department of Philosophy, University of Geneva.
      \email{dominic.ryder@unige.ch}}

\received{XX XX XXXX}
\revised{XX XX XXXX}
\accepted{XX XX XXXX}

\begin{abstract}
\citet{Manchak2018-MANPRA-5} formulate the black hole information paradox as a failure of predictability in black hole evaporation spacetimes, diagnosed by non-global hyperbolicity. I offer a strategy for resolving this paradox. I argue that failures of predictability in black hole evaporation are not well diagnosed by non-global hyperbolicity. I then consider two weakenings of global hyperbolicity: prediction and retrodiction friendliness, the failure of which could ground a new paradox. However, deidealized black hole evaporation models can be prediction and retrodiction friendly. Therefore, the information paradox cannot be based upon failures of global hyperbolicity, nor either retrodiction or prediction unfriendliness.

\end{abstract}

\maketitle

\section{Introduction}\label{sec:intro}

In recent years some philosophers, such as \cite{Belot1999hawking}, \citet{maudlin2017information} and \citet{wallace_2020}, have become sceptical about whether the black hole information paradox, construed as a failure of unitarity due to the evaporation of black holes, is truly paradoxical. This is because one can only expect globally unitary evolution on globally hyperbolic spacetimes.\footnote{A slice is a closed achronal set without an edge. A Cauchy surface is a slice such that every causal curve without an endpoint intersects it exactly once. Heuristically, a Cauchy surface registers some information about every point in spacetime. A spacetime is globally hyperbolic if it admits a Cauchy surface and thus a well-posed initial value description. A globally hyperbolic spacetime is causally well behaved.} However, the usual model of black hole evaporation fails to be globally hyperbolic. Consequently, one should not expect unitary evolution through black hole evaporation. These authors then assuage the fear that global non-unitarity violates the evolution laws of quantum field theory by appealing to formulations of quantum field theory on non-globally hyperbolic spacetimes (see \citet{kay1992principle, Yurtsever_1994, janssen2022quantum} as well as \citet[app. B]{Belot1999hawking}).

In contrast, \citet{Manchak2018-MANPRA-5} give a clear and precise statement of a paradox in which the breakdown of global hyperbolicity is precisely what is paradoxical about black hole evaporation. Global hyperbolicity is the property needed to guarantee predictability in general relativity, in the sense of admitting a well-posed initial value problem.\footnote{By predictability I mean predictability and retrodictability. When the difference becomes important I shall specify which is intended. The sense in which global hyperbolicity provides a well-posed initial value problem is given by the following theorem: given an appropriate data surface in general relativity, there exists a unique (up to isometry) spacetime that is the maximal globally hyperbolic development (MGHD) of the data surface. In general relativity an appropriate data surface is a maximal 3-dimensional smooth manifold $\Sigma$, with a smooth Riemannian metric $h_{ab}$ and a smooth symmetric tensor field $K_{ab}$ that describes the extrinsic curvature. \citet{ChoquetBruhatGeroch1969} give the statement of the uniqueness of the MGHD (see also \citet{Geroch1970}). For melancholy about determinism in general relativity, see \citet[sec. 3.8]{earman1995bangs} and \citet{Manchak2011-MANWIA}.} Manchak and Weatherall show that the restriction of physically reasonable spacetimes to those that are globally hyperbolic (a restriction common in physics practice) is inconsistent with the physical reasonableness of black hole evaporation spacetimes (given a seemingly reasonable definition of a black hole evaporation spacetime), and thus we have a paradox. Given global hyperbolicity underwrites predictability in general relativity, this paradox threatens the predictability of the laws governing black hole evaporation. This complements Hawking's original (\citeyear{Hawking1976}) discussion of information loss, which concerned the breakdown of predictability.\footnote{Here, and throughout, I use predictability as it is used in the physics literature on the information paradox: as determined by an appropriate initial data set. In a quantum theory this may mean unique specification of a final probability distribution given an initial quantum state. This is how predictability is used in \citet{Hawking1976}. Significant parts of the philosophy of physics literature understand prediction as the domain of dependence of the chronological past of some observer, i.e. what is determinable from the information contained in an observer's past light cone (e.g. \citet{geroch1977, manchak2008prediction}). This is not the sense of prediction relevant here.}

In this paper I will argue against formulating the black hole information paradox in terms of global hyperbolicity. I argue that global hyperbolicity is not the appropriate standard of predictability in black hole evaporation due to the failure of general relativity for Planck-scale physics. Thus I offer a strategy for resolving the paradox as formulated by Manchak and Weatherall. I then consider a pair of interesting weakenings of global hyperbolicity: \textit{prediction} and \textit{retrodiction friendliness}, which plausibly could provide an appropriate standard of predictability for black hole evaporation, and could thus underwrite an information paradox based upon predictability. However, I argue that our best models of black hole evaporation can be prediction and retrodiction friendly, and so no such paradox can be defended. I will now state Manchak and Weatherall's formulation of the paradox, and provide a summary of my response. 

Manchak and Weatherall state the paradox as three premises, each of which we have reason to believe, but which together are inconsistent:

\begin{enumerate}
    \item (BHE) Some evaporation spacetimes are physically reasonable.
    \item (CCH) All physically reasonable spacetimes are globally hyperbolic.
    \item (KWL) No evaporation spacetime is globally hyperbolic.
\end{enumerate}

(BHE) follows from the existence of Hawking radiation and backreaction arguments (see \citet{Hawking:1975vcx}). (CCH) is endorsed as a commitment to predictability, as global hyperbolicity is required for a well-posed initial value problem in general relativity, and ``for many physicists, (CCH) is axiomatic'' \citep[p. 4]{Manchak2018-MANPRA-5}. (KWL) follows from the Kodama-Wald-Lesourd theorem (\citet{Kodama:1979vm, wald1984black, lesourd2018causal}), which states that, assuming a seemingly reasonable definition of a black hole evaporation spacetime, black hole evaporation implies non-global hyperbolicity.\footnote{\citet{Manchak2018-MANPRA-5} use the Kodama-Wald version of the theorem \citep{Kodama:1979vm, wald1984black}, not the \citet{lesourd2018causal} extension. However, the latter is required to guarantee the breakdown of global hyperbolicity.}

In section \ref{sec:predictability} I argue that (CCH) should be rejected. Although all non-globally hyperbolic spacetimes lack a well-posed initial value problem, some such failures only indicate a bound on the domain of applicability of general relativity.\footnote{See \citet{WeatherallForthcoming-WEAWDG} for a discussion of the breakdown of general relativity.} In black hole evaporation, the failure of predictability seems to arise only from our lack of a consistent theory of Planck-scale physics that can handle the divergent curvature predicted by semi-classical gravity. Therefore, the failure of predictability that is diagnosed by non-global hyperbolicity can be reduced to our ignorance of Planck-scale physics; singularity resolution can be expected to provide an initial value problem for a black hole in the final stages of its evaporation.\footnote{\label{fn:singularity_resolution}Many attitudes can be taken to singularities in classical general relativity, only one of which is completely well-behaved singularity resolution \citep{earman1995bangs,crowther2022four}. Indeed, what constitutes singularity resolution is not stable across physics practice \citep{Thebault_2023}. In this paper, only a minimal understanding of singularity resolution --- as recovering a well-posed initial value problem --- is needed, and this minimal understanding is motivated in section \ref{sec:predictability}. For a review of singularity resolution in quantum gravity, see \citet{bojowald2007singularities}.} 

I motivate this rejection of (CCH) via an analysis of Manchak and Weatherall's  ``impressionistic'' (\citeyear[fn.9]{Manchak2018-MANPRA-5}) formulation of their paradox, which identifies a failure of predictability, not retrodictability. I discuss Manchak and Weatherall's use of a particular conformal diagram, and prove a theorem that states any spacetime which admits representation by this conformal diagram is not conformally equivalent to the only spacetime that can support the impressionistic formulation of their paradox. This leads me to argue that the breakdown of predictability identified in their paradox is a result of our ignorance of Planck-scale physics, and so the failure of predictability (not retrodictability) can be expected to be resolved by a consistent theory of Planck-scale physics.\footnote{\label{fn:aproposMaudlin}\citet{Manchak2018-MANPRA-5} devote much of their paper to astutely criticising \citet{maudlin2017information} for arguments similar to the one presented here. In service of that goal their arguments are very useful, and the response I present above turns upon considerations beyond the remit of responding to Maudlin. In particular: a) my argument explicitly appeals to a consistent theory of Planck-scale physics (usually thought of as quantum gravity), a strategy Maudlin explicitly rejects; b) I consider (in section \ref{sec:retrodictability}) models of black hole evaporation which violate the assumption of the Kodama-Wald theorem used by Manchak and Weatherall; and regardless, c) I do not claim to recover global hyperbolicity via disconnected Cauchy surfaces, which is the specific construction of Maudlin's that Manchak and Weatherall target.} Thus, we should reject global hyperbolicity as the appropriate diagnostic tool for predictability in black hole evaporation, and in turn we should reject (CCH).

However, there is a pair of spacetime properties, prediction and retrodiction friendliness, which are weaker than global hyperbolicity and which some physicists take to diagnose those failures of predictability that a consistent theory of Planck-scale physics may not be expected to resolve. These are spacetimes such that the problematic causal structure is not confined to a singular point, but is an entire spacetime region; in particular, a region isolated behind a causal horizon and that `appears' or `vanishes' (in a sense to be made precise later). That retrodiction unfriendliness is the relevant spacetime property to the information paradox is supported by the following theorem (stated in section \ref{ssec: Prediction and retrodiction friendliness} and proven in appendix \ref{app:proof}): retrodiction unfriendliness implies a pure-to-mixed transition, thus recovering the traditional information paradox. 

Having rejected (CCH), Manchak and Weatherall's paradox is resolved. However, with the weakening of global hyperbolicity to prediction and retrodiction friendliness in hand, one might attempt to resuscitate the paradox by weakening (CCH) to (*CCH):

\begin{enumerate}
    \item [*2.] (*CCH) All physically reasonable spacetimes are prediction and retrodiction friendly. 
\end{enumerate}

To complete the resuscitation, (KWL) must be strengthened to (*KWL):

 \begin{enumerate}
    \item [*3.] (*KWL) Some physically reasonable black hole evaporation spacetime is retrodiction unfriendly. 
\end{enumerate}

It is therefore interesting to consider whether we have good grounds to believe (*KWL). In section \ref{sec:retrodictability} I answer in the negative: we do not have good grounds to believe some physically reasonable black hole evaporation spacetime is retrodiction unfriendly.  According to Manchak and Weatherall's definition of a black hole evaporation spacetime, all black hole evaporation spacetimes are retrodiction unfriendly. However, the model they use to motivate this definition is highly idealized. I present deidealized models of black hole evaporation which are retrodiction friendly. I argue these models are more reasonable than the idealized models usually discussed in the literature, and thus that physically reasonable black hole evaporation spacetimes may be retrodiction friendly. Therefore, one should not define black hole evaporation spacetimes such that they are retrodiction unfriendly --- in contrast to the definition used by Manchak and Weatherall --- and so the resuscitated paradox is avoided. I conclude that neither global hyperbolicity nor retrodiction friendliness can underwrite an information paradox: black hole evaporation may well be prediction and retrodiction friendly.

\section{The Information Paradox Requires Retrodiction Unfriendliness}\label{sec:predictability}

Manchak and Weatherall defend their formulation of the paradox via a discussion of an ``impressionistic'' \citep[fn.9]{Manchak2018-MANPRA-5} version. In this section I present this impressionistic version and argue that the failure of predictability (not retrodictability) therein is just the result of our ignorance of Planck-scale physics. This motivates a rejection of (CCH). I then introduce two weaker spacetime properties than global hyperbolicity, prediction and retrodiction friendliness, which capture an asymmetry between predictability and retrodictability in black hole evaporation; evaporation-Schwarzschild is prediction friendly but retrodiction unfriendly. I suggest that failures of retrodictability due to retrodiction unfriendliness may go beyond our ignorance of Planck-scale physics, and that it is reasonable to doubt that a consistent theory in that regime will recover retrodictable evolution for such spacetimes. 

\subsection{Planck-Scale Physics and Predictability}\label{ssec:Planckscale physics}

The impressionistic version of the paradox relies on two conformal diagrams. Roughly, conformal diagrams depict the causal structure of a spacetime, whilst shrinking infinity into a finite distance. The first diagram, figure \ref{fig:CollapseSchwarzschild}, is a model of an uncharged, non-rotating black hole which forms by the collapse of stellar matter and then remains for all time. I call this \textit{collapse-Schwarzschild}. The second diagram, figure \ref{fig:EvaporationSchwarzschild}, is a model of an uncharged, non-rotating black hole, which also forms by the collapse of stellar matter, but then evaporates and eventually disappears, leaving behind a region of spacetime isometric to flat Minkowski spacetime. I call this \textit{evaporation-Schwarzschild}.

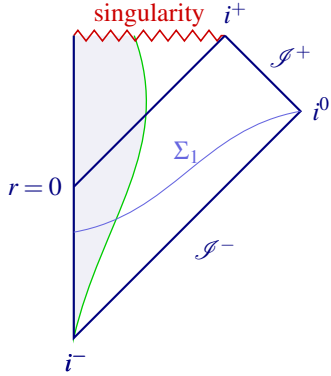
\begin{figure}
    \centering
    \begin{tikzpicture}[scale=2]
  \message{collapse-Schwarzschild}
  \coordinate (O) at (0, 0); 
  \coordinate (NE)  at ( 1, 1); 
  \coordinate (S)  at ( 0,-1); 
  \coordinate (N)  at ( 0, 1); 
  \coordinate (E)  at ( 1.5, 0.5); 
  \coordinate (X)  at ( 0.4, 1); 
  \coordinate (Y) at (0, -0.3); 

  \draw[particle, fill = mylightblue]
      (N) to (S) to[out=77,in=-70] (X);
  
  \draw[singularity] (N) -- node[above] {singularity} (NE);
  
  \draw[thick,mydarkblue] (NE) -- (E) -- (S) -- (N);
  \draw[thick,mydarkblue] (O) -- (NE);


  \node[above=0,left=1,mydarkblue] at (O) {$r=0$};
  \node[above=1,right=1,mydarkblue] at (E) {$i^0$};
  \node[right=1,below=1,mydarkblue] at (S) {$i^-$};
  \node[right=4,above=1,mydarkblue] at (NE) {$i^+$};
  \node[right=1,below=1,mydarkblue] at (S) {$i^-$};

  \node[mydarkblue,above right=-1] at (1.25,0.75) {$\If$};
  \node[mydarkblue,below right=-1] at (0.75,-0.3) {$\Ip$};

\tikzset{mylabel/.style  args={at #1 #2  with #3}{
    postaction={decorate,
    decoration={
      markings,
      mark= at position #1
      with  \node [#2] {#3};
 } } } }
  
  \draw[world line]
      (Y) to [out = 10, in = 190] node[midway, above] {$\Sigma_1$} (E) ;
  
\end{tikzpicture}
    \caption{The conformal diagram for stellar collapse into a Schwarzschild black hole. The shaded region represents matter undergoing collapse.}
    \label{fig:CollapseSchwarzschild}
\end{figure}

\begin{figure}
    \centering
    \begin{tikzpicture}[scale=2.3]
  \message{evaporation-Schwarzschild}
  
  \coordinate (O) at (0, 0); 
  \coordinate (NE)  at (0.5, 0.5); 
  \coordinate (NN) at (0.5, 1); 
  \coordinate (S)  at ( 0,-1); 
  \coordinate (N)  at ( 0, 0.5); 
  \coordinate (E)  at ( 1.25, 0.25); 
  \coordinate (X)  at ( 0.2, 0.5); 
  \coordinate (Y1) at (0, -0.3); 
  \coordinate (Y2) at (0.5, 0.6);
  
  \draw[particle, fill = mylightblue]
      (N) to (S) to[out=77,in=-70] (X);
  
  \draw[singularity] (N) -- node[above] {} (NE);
  
  \draw[thick,mydarkblue] (NE) -- (NN) -- (E) -- (S) -- (N);
  \draw[thick,mydarkblue] (O) -- (NE);

    \path (O) -- (NE);

  \node[above=0,left=1,mydarkblue] at (O) {$r=0$};
  \node[above=1,right=1,mydarkblue] at (E) {$i^0$};
  \node[right=1,below=1,mydarkblue] at (S) {$i^-$};
  \node[right=1,above=1,mydarkblue] at (NN) {$i^+$};
  \node[right=1,below=1,mydarkblue] at (S) {$i^-$};

  \node[mydarkblue,above right=-1] at (0.75,0.75) {$\If$};
  \node[mydarkblue,below right=-1] at (0.75,-0.3) {$\Ip$};


\tikzset{mylabel/.style  args={at #1 #2  with #3}{
    postaction={decorate,
    decoration={
      markings,
      mark= at position #1
      with  \node [#2] {#3};
 } } } }
  
  \draw[world line]
      (Y1) to [out = 10, in = 190] node[midway, above] {$\Sigma_1$} (E) ;

    \draw[world line]
      (Y2) to [out = -10, in = 170] node[pos=0.3, above] {$\Sigma_2$} (E) ;
  
\end{tikzpicture}
    \caption{The conformal diagram representing the evaporation of a Schwarzschild black hole formed by collapse. The mass of the black hole is shrinking over time, and after the evaporation, the spacetime is locally Minkowski. Neither $\Sigma_1$ nor $\Sigma_2$ is a Cauchy surface.}
    \label{fig:EvaporationSchwarzschild}
\end{figure}
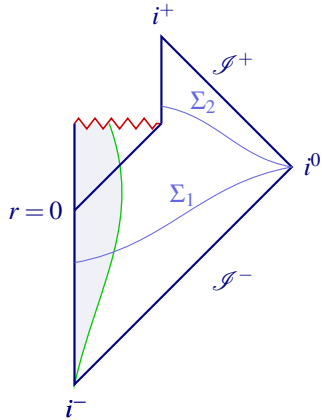

Manchak and Weatherall's impressionistic version of the information loss paradox is the following\footnote{I have replaced Manchak and Weatherall's figure references with my own.}:

\begin{quote}
   Focus attention on just the bottom half of Fig. [\ref{fig:EvaporationSchwarzschild}], as depicted in Fig. [\ref{fig:CollapseSchwarzschild}]\dots there should be a maximal spacetime that we get by allowing $\Sigma_1$ to evolve according to the laws — and indeed, its Penrose diagram should look like Fig. [\ref{fig:CollapseSchwarzschild}]. But the spacetime in Fig. [\ref{fig:CollapseSchwarzschild}] is, by construction, extendible, in the sense that there is a proper isometric embedding of this spacetime into the one depicted in Fig. [\ref{fig:EvaporationSchwarzschild}]\dots it follows that the laws, plus initial data specified on $\Sigma_1$, do not determine what happens indefinitely into the future, because this second spacetime is not globally hyperbolic, and $\Sigma_1$ fails to be a Cauchy surface. One reaches a horizon across which the laws no longer determine what happens. This situation suggests that something has gone seriously wrong. Hence the paradox. \citep[pp. 6-7]{Manchak2018-MANPRA-5}
\end{quote}

This paradox turns on an expectation of a well-posed initial value problem in general relativity. It also involves a failure of predictability, not retrodictability. The failure putatively deserves the epithet `paradox' because all physical theories we have discovered to date admit a well-posed initial value problem.\footnote{I ignore here issues of determinism pertaining to the measurement problem. See also \citet{earman1986primer} for further hesitations that I will ignore.} Predictability may fail, but given the success of the paradigm of well-posed initial value problems, such a failure would amount to learning something extremely deep and surprising about the physical world.\footnote{I encourage readers to recall that here predictability is essentially being elided with determinism, in that it means something like: admits a well-posed initial value problem, as discussed in the introduction.}

However, the failure of predictability is not such a great threat when it is merely a result of our ignorance of Planck-scale physics. To see this, consider what Manchak and Weatherall call the ``bottom half'' of evaporation-Schwarzschild. They claim (in the caption of their Fig. 3, p. 7) that the bottom half of evaporation-Schwarzschild is a spacetime conformally equivalent to collapse-Schwarzschild, and thus represented by the same conformal diagram. However, as I will now show, no plausible interpretation of Manchak and Weatherall's ``bottom half'' can both: I) support their impressionistic version of the information paradox, and II) be represented with the same conformal diagram as collapse-Schwarzschild.

The two plausible candidate spacetimes that can be embedded into evaporation-Schwarzschild prior to the evaporation time (i.e. plausible candidates for what Manchak and Weatherall call the ``bottom half'') are: (a) the causal past of the black hole region, $J^-(\mathcal{B})$, and (b) the MGHD of $\Ip$, $D(\Ip)$.\footnote{Spacetime (b) should really be defined as the MGHD of some maximal achronal slice in the interior of the spacetime which does not intersect the Cauchy horizon, as the points of $\Ip$ are not in the physical spacetime and timelike curves reach $i^-$ rather than $\Ip$. However, I will continue to refer to $D(\Ip)$ as a convenient abuse of notation.} These are the regions below the dashed lines in figure \ref{fig:bothEmbeddings} and are drawn as standalone spacetimes in figure \ref{fig:CausalPastEmbed}. 

\begin{figure}
    \centering
    \begin{tikzpicture}[scale=2.5]
 
  \coordinate (O) at (0, 0); 
  \coordinate (NE)  at (0.5, 0.5); 
  \coordinate (NN) at (0.5, 1); 
  \coordinate (S)  at ( 0,-1); 
  \coordinate (N)  at ( 0, 0.5); 
  \coordinate (E)  at ( 1.25, 0.25); 
  \coordinate (X)  at ( 0.2, 0.5); 
  \coordinate (Y1) at (0.75,0.75); 
  \coordinate (Y2) at (1, 0);
  
  \draw[particle, fill = mylightblue]
      (N) to (S) to[out=77,in=-70] (X);
  
  \draw[singularity] (N) -- node[above] {} (NE);
  
  \draw[thick,mydarkblue] (NE) -- (NN) -- node[midway, above right, mydarkblue] {$\If$} (E) -- node[midway, below right, mydarkblue] {$\Ip$} (S) -- (N);
  \draw[thick,mydarkblue] (O) -- (NE);

    \path (O) -- (NE);

  \node[above=0,left=1,mydarkblue] at (O) {$r=0$};
  \node[above=1,right=1,mydarkblue] at (E) {$i^0$};
  \node[right=1,below=1,mydarkblue] at (S) {$i^-$};
  \node[right=1,above=1,mydarkblue] at (NN) {$i^+$};
  \node[right=1,below=1,mydarkblue] at (S) {$i^-$};


\tikzset{mylabel/.style  args={at #1 #2  with #3}{
    postaction={decorate,
    decoration={
      markings,
      mark= at position #1
      with  \node [#2] {#3};
 } } } }

    \draw[dashed]
      (NE) to node[pos=0.5, below right = -0.1cm] {(b)} (Y1) ;
\draw[dashed]
      (NE) to node[pos=0.5, below left = -0.1cm] {(a)} (Y2) ;
  
\end{tikzpicture}
    \caption{The region below (a) is the causal past of the black hole, and the region below (b) is the MGHD of $\Ip$.}
    \label{fig:bothEmbeddings}
\end{figure}
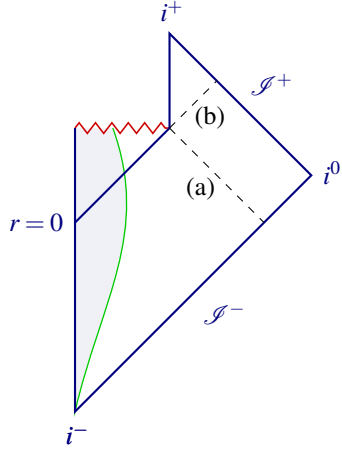

\begin{figure}

\centering

 \begin{minipage}[h]{.45\textwidth}
   \begin{tikzpicture}[scale=2.5]
  \message{Causal past}

  \coordinate (O) at (0, 0); 
  \coordinate (NE)  at (0.5, 0.5); 
  \coordinate (NN) at (0.5, 1); 
  \coordinate (S)  at ( 0,-1); 
  \coordinate (N)  at ( 0, 0.5); 
  \coordinate (E)  at ( 1.25, 0.25); 
  \coordinate (X)  at ( 0.2, 0.5); 
  \coordinate (Y1) at (0.75,0.75); 
  \coordinate (Y2) at (1, 0);
  
  \draw[particle, fill = mylightblue]
      (N) to (S) to[out=77,in=-70] (X);
  
  \draw[singularity] (N) -- node[above] {} (NE);
  
  \draw[thick,mydarkblue] (Y2) -- node[midway, below right, mydarkblue] {$\Ip$}  (S) -- (N);
  \draw[thick,mydarkblue] (O) -- (NE);

    \path (O) -- (NE);

  \node[above=0,left=1,mydarkblue] at (O) {$r=0$};
  \node[right=1,below=1,mydarkblue] at (S) {$i^-$};
  \node[right=1,below=1,mydarkblue] at (S) {$i^-$};

\tikzset{mylabel/.style  args={at #1 #2  with #3}{
    postaction={decorate,
    decoration={
      markings,
      mark= at position #1
      with  \node [#2] {#3};
 } } } }

      \draw[dashed]
      (NE) to node[midway, above right = -0.1cm, mydarkblue] {$\delta J^-(\mathcal{B})$} (Y2) ;
  
\end{tikzpicture}
\end{minipage}
\begin{minipage}[h]{.45\textwidth}
\begin{tikzpicture}[scale=2.5]
    \coordinate (O) at (0, 0); 
  \coordinate (NE)  at (0.5, 0.5); 
  \coordinate (NN) at (0.5, 1); 
  \coordinate (S)  at ( 0,-1); 
  \coordinate (N)  at ( 0, 0.5); 
  \coordinate (E)  at ( 1.25, 0.25); 
  \coordinate (X)  at ( 0.2, 0.5); 
  \coordinate (Y1) at (0.75,0.75); 
  \coordinate (Y2) at (1, 0);
  
  \draw[particle, fill = mylightblue]
      (N) to (S) to[out=77,in=-70] (X);
  
  \draw[singularity] (N) -- node[above] {} (NE);
  
  \draw[thick,mydarkblue] (Y1) -- (E) -- (S) -- (N);
  \draw[thick,mydarkblue] (O) -- (NE);

    \path (O) -- (NE);

  \node[above=0,left=1,mydarkblue] at (O) {$r=0$};
  \node[above=1,right=1,mydarkblue] at (E) {$i^0$};
  \node[right=1,below=1,mydarkblue] at (S) {$i^-$};
  \node[right=1,below=1,mydarkblue] at (S) {$i^-$};

  \node[mydarkblue,above right=-1] at (0.95,0.55) {$\If$};
  \node[mydarkblue,below right=-1] at (0.75,-0.3) {$\Ip$};


\tikzset{mylabel/.style  args={at #1 #2  with #3}{
    postaction={decorate,
    decoration={
      markings,
      mark= at position #1
      with  \node [#2] {#3};
 } } } }

    \draw[dashed]
      (NE) to node[pos=0.5, above left = -0.1cm, mydarkblue] {$\mathcal{H}^C$} (Y1) ;

\end{tikzpicture}
    \end{minipage}
    \caption{Left: The casual past of the black hole, $J^-(\mathcal{B})$. Right: The MGHD of $\Ip$, $D(\Ip)$.}
    \label{fig:CausalPastEmbed}
\end{figure}
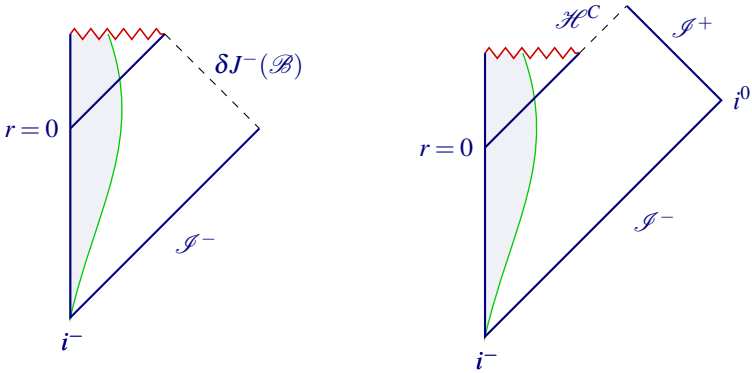

Consider first a) the causal past of the black hole; this spacetime violates conjunct I: it cannot support the impressionistic version of the paradox. The causal past of the black hole is extendible, but its extendibility does not signal a breakdown of predictable evolution because it is merely a result of restricting to a submanifold of a globally hyperbolic spacetime, and there is clearly no problem with predictable evolution in the extended globally hyperbolic spacetime. To see this, first note that, although the causal past of the black hole is conformally equivalent to collapse-Schwarzschild, it has very different asymptotic structure: the future null boundary is neither asymptotically flat nor infinitely far away. Thus, the label `future null infinity' for this boundary of $J^-(\mathcal{B})$ is somewhat misleading.\footnote{\citet[p. 272]{birrell_davies_1982} similarly use `future null infinity' to describe the finitely far away strongly curved null boundary of the causal past of the black hole, so Manchak and Weatherall may be following this textbook usage.} Importantly, this obscures the fact that the extendibility of the spacetime is trivial in the sense that we have taken a submanifold of a globally hyperbolic spacetime, and so we can extend back to this globally hyperbolic spacetime and no problem with predictable evolution arises.

In more detail, the paradox requires our expectation of predictable laws to be ruined by the fact that the globally hyperbolic ``bottom half'' spacetime can be embedded into the larger evaporation-Schwarzschild spacetime which is not determined by a Cauchy surface for the former. However, the causal past of the black hole can be embedded into the globally hyperbolic MGHD of $\Ip$. Therefore, we already have the result that a Cauchy surface for the ``bottom half'' spacetime (construed for now as the causal past of the black hole) does not ``determine what happens indefinitely into the future'' \citep[p. 6]{Manchak2018-MANPRA-5} because a Cauchy surface for the causal past of the black hole is not a Cauchy surface for the MGHD of $\Ip$. However, this is not a paradox; it is just the result of restricting to a submanifold of a globally hyperbolic spacetime which does not contain a Cauchy surface for that spacetime. One would get the exact same extendibility by taking a partial Cauchy surface in Minkowski spacetime, and Minkowski spacetime clearly has no predictable evolution problem. Therefore, the causal past of the black hole cannot support the impressionistic version of the paradox. Hence, the only paradox-supporting candidate to be the ``bottom half'' of evaporation-Schwarzschild is the MGHD of $\Ip$. 

Consider then b) the MGHD of $\Ip$; this spacetime violates conjunct II: it cannot be represented with the same conformal diagram as collapse-Schwarzschild. In appendix \ref{app:proof} I prove the following theorem:

\begin{restatable}{theorem}{cfinequiv}\label{thm:cfinequiv}

    Let $(M,g)$ be the spacetime described by the MGHD of $\Ip$ and let $(N,h)$ be collapse-Schwarzschild. There does not exist a conformal isometry $\psi: M \rightarrow N$.
\end{restatable} 

In sketch: there is a region of the MGHD of $\Ip$ which is not in the black hole region and from which the black hole region is inaccessible by a causal worldline; but no such region exists in collapse-Schwarzschild, and because conformal transformations do not change causal structure, the two spacetimes cannot be conformally equivalent. 

Therefore, neither a) the causal past of the black hole, nor b) the MGHD of $\Ip$, can both I) support the impressionistic version of the paradox, and II) be represented with the same conformal diagram as collapse-Schwarzschild. Why does this matter? Because the ``bottom half'' conformal diagram is used as an intuition pump to the conclusion that something ``has gone seriously wrong'' with predictability in black hole evaporation. Collapse-Schwarzschild is bounded to the future by the singularity and future null infinity. Therefore, if collapse-Schwarzschild were the ``maximal spacetime that we get by allowing $\Sigma_1$ [a Cauchy surface for the ``bottom half''] to evolve according to the laws'' as Manchak and Weatherall state, it would be very surprising that this spacetime can be embedded into a larger, non-globally hyperbolic spacetime. However, the paradox-supporting ``bottom half'' spacetime, the MGHD of $\Ip$, is bounded to the future by a Cauchy horizon induced by the singularity (unlike the future null boundary of the causal past which is a Cauchy horizon induced by using a partial Cauchy surface). Given this singularity-induced Cauchy horizon, it is not at all surprising that the MGHD of $\Ip$ can be embedded into a larger spacetime, and that this larger spacetime is not globally hyperbolic. 

More importantly, we know what gives rise to the predictability-breaking Cauchy horizon: the naked singularity due to the evaporation event. Consider again figure \ref{fig:bothEmbeddings}: evolution from $\Ip$ will be completely predictable until a slice just prior to the Cauchy horizon when, at the evaporation time, the singularity is exposed and ruins the predictability of the evolution. Therefore, if we had a consistent theory of Planck-scale physics, we could state what happens at the evaporation event and this would be sufficient to recover predictable evolution through the evaporation event.\footnote{Of the four attitudes to singularities of \citet{crowther2022four}, the assumption that a well-posed initial value problem will be recovered in singularity resolution is consistent with the first, second and fourth, which are the better-motivated attitudes given independent arguments that general relativity is merely a low-energy effective theory. It will also be consistent with the third attitude if singularities and a well-posed initial value problem can be made to co-exist, but this seems unlikely.}

This response to the impressionistic version of the paradox captures the physicist intuition that evaporation events are not a barrier to deterministic evolution into the future. Wald, for example, makes this intuition explicit: the failure of $\Sigma_1$ in figure \ref{fig:EvaporationSchwarzschild} to be a Cauchy surface: 

\begin{quote}
    ``\dots appears to result from only the `single missing point' corresponding to the final instant of black-hole evaporation. This does not appear to be a serious obstacle to obtaining deterministic evolution in quantum gravity.'' \citep[fn. 15]{Wald1980}
\end{quote}
See also \citet{wald1984black} for a sketched formal argument defending this intuition. Similarly, in his original information paradox paper, \citet{Hawking1976} is unperturbed by the naked singularity and the associated breakdown of predictable evolution.\footnote{The precise technical formulation of a well-posed initial value problem in such a theory, whether unitary evolution, conditions on superscattering operators, or something else, will have to wait for a consistent theory of Planck-scale physics. But I follow Hawking and Wald in presuming such a formulation will be provided.}

This suggests that --- in the context of black hole evaporation --- we should reject (CCH) in Manchak and Weatherall's paradox; that is, reject the claim that all physically reasonable spacetimes are globally hyperbolic. Black hole evaporation takes us out of the domain in which we expect general relativity to be valid, and so general relativity is not the appropriate theory in which to formulate a well-posed initial value problem. The failure of predictability signaled by the failure of global hyperbolicity is just our ignorance of the consistent theory of Planck-scale physics, rather than an interesting consequence of the appropriate physical theory on the domain. Thus, we resolve Manchak and Weatherall's statement of the paradox by rejecting (CCH).

\subsection{Prediction and Retrodiction Friendliness}\label{ssec: Prediction and retrodiction friendliness}

The arguments of the previous section concerned the ability of Planck-scale physics to recover future-orientated predictable evolution through the evaporation event. However, the failures of predictability in black hole evaporation often concern failures of retrodictability, not future-orientated predictability. Manchak and Weatherall are aware of this asymmetry, writing ``the puzzle concerns whether any specification of data on a particular surface is sufficient to retrodict the physical process by which that data came about'' (\citeyear[p. 1]{Manchak2018-MANPRA-5}). It is therefore a further question whether a consistent theory of Planck-scale physics can be expected to recover retrodictable evolution through black hole evaporation. It is possible that, whereas the failure of predictability can be reasonably reduced to the problem of Planck-scale physics, the failure of retrodictability cannot and thus generates an information paradox.

One response to the failure of retrodictability in black hole evaporation is to import the arguments of the previous section \textit{mutatis mutandis}. Black hole evaporation takes us outside of the domain of applicability of general relativity and so, plausibly, as with predictability, we can expect a consistent theory of Planck-scale physics to recover retrodictable evolution. Someone convinced by this argument, applied to retrodiction, will then be satisfied that there is no paradox due to retrodictability. However, as discussed above, predictability and retrodictability are usually treated asymmetrically in the black hole information paradox. Moreover, the same physicists who were confident a consistent theory of Planck-scale physics would recover predictable evolution are more sceptical in the case of retrodiction. For example, continuing the quote above, Wald writes:

\begin{quote}
    ``On the other hand, intuitively, [$\Sigma_2$ in figure \ref{fig:EvaporationSchwarzschild}] does not come close to being a Cauchy surface for development into the past'' \citep[fn. 15]{Wald1980}.
\end{quote} 
\citet[p. 2462]{Hawking1976} too shares this intuition:

\begin{quote}
    ``Measurements at future infinity are insufficient to determine completely the state of the system at past infinity: One also needs data on the event horizon describing what fell into the black hole.''
\end{quote}



The worry is that retrodiction requires not just a description of the singularity, but everything that fell into the black hole, and this information is strictly isolated from future null infinity due to the event horizon. The failure of prediction is just due to a naked singularity, whereas the failure of retrodiction is due to the existence of a region of spacetime causally isolated from future null infinity that `vanishes'. So although the failure of predictability is reducible to our ignorance of Planck-scale physics, the standard view is that failures of retrodictability due to the `vanishing' of regions of spacetime causally isolated from future null infinity are not so reducible and lead to an information paradox. In this section I introduce two  novel spacetime properties, prediction and retrodiction friendliness, which can parse the asymmetry between prediction and retrodiction in black hole evaporation spacetimes.


The talk of vanishing spacetime regions above lacks conceptual clarity, and will be made precise shortly. For now, the foregoing discussion suggests the following crudely stated spacetime properties: a spacetime is \textit{retrodiction friendly} if it does not have a region causally isolated from future null infinity that `disappears'; a spacetime is retrodiction unfriendly if it has such a region. Likewise, a spacetime is \textit{prediction friendly} if it does not have a region of spacetime causally isolated from past null infinity that `appears'; a spacetime is prediction unfriendly if it has such a region.

Minkowski spacetime with an arbitrary point removed fails to be globally hyperbolic but is prediction and retrodiction friendly: no region of the spacetime is ever causally isolated from future null infinity.\footnote{More physically reasonable examples of retrodiction friendly but non-globally hyperbolic spacetimes will be given in section \ref{sec:retrodictability}, motivated by black hole evaporation.} On the other hand, evaporation-Schwarzschild --- in addition to failing to be globally hyperbolic --- fails to be retrodiction friendly because the black hole region is precisely one which is causally isolated from future null infinity and then `vanishes'. 

\begin{figure}
    \centering
    \begin{tikzpicture}[scale=3]
  \message{Retrodiction friendliness}
  
  \coordinate (O) at (0, 0); 
  \coordinate (NE)  at (0.5, 0.5); 
  \coordinate (NN) at (0.5, 1); 
  \coordinate (S)  at ( 0,-1); 
  \coordinate (N)  at ( 0, 0.5); 
  \coordinate (E)  at ( 1.25, 0.25); 
  \coordinate (X)  at ( 0.2, 0.5); 
  \coordinate (Y1) at (0, 0.2); 
  \coordinate (Y2) at (0.5, 0.6);

\clip[]
      (-0.05, 0) |- (NN) -| (E) |-++ (-1.3,-0.5) |- cycle;
  
  \draw[singularity] (N) -- node[above] {} (NE);
  
  \draw[thick,mydarkblue] (NE) -- (NN) -- (E) -- (S) -- (N);
  \draw[thick,mydarkblue] (O) -- (NE);

    \path (O) -- (NE);

  \node[above=0,left=1,mydarkblue] at (O) {};
  \node[above=1,right=1,mydarkblue] at (E) {};
  \node[right=1,below=1,mydarkblue] at (S) {};
  \node[right=1,above=1,mydarkblue] at (NN) {};
  \node[right=1,below=1,mydarkblue] at (S) {};

  \node[mydarkblue,above right=-1] at (0.75,0.75) {};
  \node[mydarkblue,below right=-1] at (0.75,-0.3) {};

\tikzset{mylabel/.style  args={at #1 #2  with #3}{
    postaction={decorate,
    decoration={
      markings,
      mark= at position #1
      with  \node [#2] {#3};
 } } } }
  
  \draw[world line, mylabel=at 0.5 above left with {}]
      (Y1) to [out = 10, in = 190] (E) ;

    \draw[world line, mylabel=at 0.5 below with {}]
      (Y2) to [out = -10, in = 170] (E) ;

      \draw[dashed]
      (NE) --++ (0.28, -0.28) ;
      
      \draw[densely dotted]
      (NE) --++ (0.08, 0.08) ;

\draw [decorate,decoration={brace,amplitude=5pt}]
  (0.5, 0.62) -- (0.6, 0.6) node[midway,yshift=0.8em, xshift=0.1em]{$\scriptstyle Q$};

  \draw [decorate,decoration={brace,amplitude=5pt}]
  (0, 0.24) -- (0.3, 0.24) node[midway,yshift=0.8em]{$\scriptstyle V$};

\draw [decorate,decoration={brace,amplitude=5pt,mirror}]
  (0, 0.2) -- (0.78, 0.2) node[midway,yshift=-0.8em]{$\scriptstyle K$};

   \node at (0.27, 0.42)[circle,fill,inner sep=1.5pt]{};
    \node at (0.31, 0.38)[] {$\scriptstyle p$};
  
\end{tikzpicture}
    \caption{The regions $K$, $V$ and $Q$ in the definition of retrodiction (un)friendliness for evaporation-Schwarzschild. $p$ is a point within $D^+(\Sigma_1)- J^+(\Sigma_2)\cup J^-(\Sigma_2)$, violating condition (iv) of the definition. Thus the spacetime is retrodiction unfriendly.}
    \label{fig:Retrodictionunfriendly}
\end{figure}
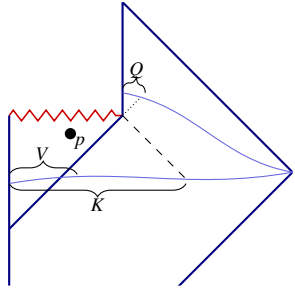

A formal definition may be stated as follows:

\begin{definition}[Retrodiction Friendly]
    Let $(M,g)$ be a non-totally imprisoning spacetime with connected, achronal, edgeless sets $\Sigma_1, \Sigma_2 \subset I^+(\Sigma_1)$ such that $\Sigma_1$ is closed and non-compact, and (i) $J^+(K)\cap \Sigma_2$ has compact closure, where $K = \Sigma_1 - D^-(\Sigma_2)$, (ii) $J^-(Q) \cap \Sigma_1$ has compact closure where $Q=\Sigma_2 - D^+(\Sigma_1)$ , and (iii) $\Sigma_1 - V \subset J^-(\Sigma_2)$ where $V$ is a compact subset of $\Sigma_1$. This spacetime is \textit{retrodiction friendly} if, for any such $\Sigma_1$ and $\Sigma_2$, (iv) $D^+(\Sigma_1)- J^+(\Sigma_2)\cup J^-(\Sigma_2) = \emptyset$. The spacetime is \textit{retrodiction unfriendly} if and only if it is not retrodiction friendly.
\end{definition}

Figure \ref{fig:Retrodictionunfriendly} depicts the regions used in the above definition for evaporation-Schwarzschild, along with a point within $D^+(\Sigma_1)- J^+(\Sigma_2)\cup J^-(\Sigma_2)$, violating condition (iv). The definition of prediction friendliness is the time reverse of the definition of retrodiction friendliness.

Retrodiction unfriendliness is intended to identify spacetimes in which retrodictability fails due to black hole evaporation. Conditions (i), (ii), and (iii) ensure that any failures of predictability are due to the black hole evaporation. Conditions (i) and (ii) use regions $K$ and $Q$, which are the subsets of $\Sigma_1$ and $\Sigma_2$ that are not determined by the other achronal slice, and thus represent the failures of predictability. The compactness conditions associated with the causal developments of these regions ensure the failures of predictability are confined to a finite region, because black hole evaporation occurs in a finite region. For example, anti-de Sitter spacetime without reflecting boundary conditions fails to be predictable because information can come in from infinity. Conditions (i) and (ii) rule out anti-de Sitter because its failure of predictability is due to asymptotic spacetime structure, rather than due to an evaporating black hole. Condition (iii) is needed because \citet{lesourd2018causal} found examples of spacetimes that satisfy conditions (i) and (ii), and that violate condition (iv), but for which the violation is not due to a point being causally isolated behind a singularity, and indeed bear no resemblance to black hole evaporation. These examples are such that $J^-(\Sigma_1) \cap \Sigma_2 = \emptyset$ and $J^-(\Sigma_2) \cap \Sigma_1 = \emptyset$, so neither slice can be said to be in the future of the other, a circumstance that does not arise due to black hole evaporation. Condition (iii) disqualifies spacetimes such as these. 

Condition (iv) encodes the fact that there is no region causally isolated from $\Sigma_2$ (an achronal slice in the future) but determined by $\Sigma_1$ (an achronal slice in the past). This feature of evaporation-Schwarzschild --- a region causally isolated from a future achronal slice --- is the asymmetric causal structure that suggests retrodiction, not prediction, fails in black hole evaporation, as discussed above.
 
Checking a few more standard spacetimes: globally hyperbolic spacetimes are prediction and retrodiction friendly; Minkowski spacetime with a point removed is also; maximally extended Kerr and Reissner-Nordstr\"om are also;  Minkowski spacetime with a compact spacelike 3-surface removed is prediction and retrodiction unfriendly; finally, evaporation-Schwarzschild is retrodiction unfriendly but prediction friendly.\footnote{Note $\Sigma_1$ and $\Sigma_2$ must be connected, so retrodiction friendliness means retrodictability from a single connected component of the Universe. Therefore, baby universe models are retrodiction unfriendly.}


That retrodiction unfriendliness is the appropriate spacetime property to attempt to ground an information paradox is further supported by the fact that retrodiction unfriendliness implies a pure-to-mixed state transition. A formal statement of the mixedness of a state on $\Sigma_2$ in a retrodiction unfriendly spacetime (given a causally well-behaved past) is the following, which I sketch a proof of in appendix \ref{app:proof}\footnote{I prove this theorem as a generalisation of an argument in \citet[sec. 3]{Belot1999hawking} that evaporation-Schwarzschild implies pure-to-mixed state transition. The proof of \citet{Belot1999hawking} is sufficient to establish the existence of a model of black hole evaporation that implies a pure-to-mixed state transition. However, a further question can be asked: `which black hole evaporation models will result in a pure-to-mixed state transition?' The theorem provided here helps to answer this further question.}:

\begin{restatable}{theorem}{mixed}\label{thm:mixedness}

    Let $(M,g)$ be a retrodiction unfriendly spacetime, let $\mathcal{W}$ be a C*-algebra of observables associated with $J^+(\Sigma_1)\cup J^-(\Sigma_1)$ and $\omega$ the state on $\mathcal{W}$. If the following postulates hold:
    \begin{enumerate}
        \item $J^-(\Sigma_1)$ is such that for any subregions $S_1,S_2\subset \Sigma_1$, $J^-(S_1)\cap J^-(S_2) \neq \emptyset$.
        \item If two regions are spacelike separated, then their algebras of observables commute.
        \item If two spacelike separated regions $R_1$ and $R_2$ are such that $J^-(R_1)\cap J^-(R_2) \neq \emptyset$, then their associated algebras of observables $\mathcal{A}_1$ and $\mathcal{A}_2$ contain some elements $a_1\in \mathcal{A}_1$, $a_2\in \mathcal{A}_2$ such that $\omega(a_1a_2) \neq \omega(a_1)\omega(a_2)$.
    \end{enumerate}

    then $\omega_{\Sigma_2}$, the restriction of $\omega$ to the algebra of observables associated with $\Sigma_2$, is mixed. 
\end{restatable} 

Postulate 1 assumes that any two subregions of $\Sigma_1$ share a common past. Postulate 2 is a statement of microcausality, an axiom of axiomatic QFT.\footnote{AQFT is normally restricted to globally hyperbolic spacetimes; however, it can be extended to non-globally hyperbolic spacetimes (see \citet{kay1992principle, Yurtsever_1994, janssen2022quantum}). The core of the proof --- purity of a state implies no-correlation between commuting observables, or conversely, correlation between commuting observables implies mixedness --- should go through unproblematically. Proving this explicitly in the cited frameworks would be interesting, but given this proof is not necessary for the argument of this paper I do not carry out this check here.} Postulate 3, in the present context of Hawking radiation, encodes the entanglement between interior and exterior Hawking modes. More generally, it is supported by the genericity of vacuum correlations in QFT.

A corollary is then that if the state on $\Sigma_1$ is pure, then the evolution from $\Sigma_1$ to $\Sigma_2$ is a pure-to-mixed state transition, thus recovering the traditional formulation of the information paradox.\footnote{This argument does not necessarily imply that unitarity is the appropriate formulation of an initial value problem in some future theory of Planck-scale physics, but only further motivates the claim that from the semi-classical perspective, if there is a spacetime property relevant to the information paradox, it is retrodiction unfriendliness.} 

Evaporation-Schwarzschild is prediction friendly but retrodiction unfriendly. Because prediction and retrodiction friendliness are weaker than global hyperbolicity, it is possible to reject (CCH), but instead demand that all physically reasonable spacetimes are prediction and retrodiction friendly. A new question then arises: can an information paradox, analogous to that of Manchak and Weatherall, be given that is based upon failures of retrodiction friendliness? I consider this question in the next section. I argue that, if one wishes to restrict physically reasonable spacetimes to retrodiction friendly spacetimes, there is still no paradox because it is plausible that all physically reasonable black hole evaporation spacetimes are retrodiction friendly.


\section{Is Black Hole Evaporation Retrodiction Friendly?}\label{sec:retrodictability}

In this section I consider an analogous information paradox to that given by Manchak and Weatherall, formulated in terms of retrodiction friendliness. I state this new formulation in section \ref{ssec:resuscitating}. Then, in section \ref{ssec:BHEcanbeRetrodictionFriendly}, I argue that the premise analogous to the Kodama-Wald-Lesourd theorem in the new formulation is not well-motivated, thus undermining the grounds for the paradox.

\subsection{Resuscitating the Paradox}\label{ssec:resuscitating}

Recall Manchak and Weatherall's precise formulation of the information paradox, given in section \ref{sec:intro}. In section \ref{ssec:Planckscale physics} I argued that non-globally hyperbolic spacetimes are physically reasonable if one can expect a consistent theory of Planck-scale physics to `fix' the breakdown of predictability. Thus, we can resolve Manchak and Weatherall's paradox by rejecting (CCH) as too strong. However, in section \ref{ssec: Prediction and retrodiction friendliness}, I considered two weaker spacetime properties than global hyperbolicity: prediction and retrodiction friendliness. Violations of retrodiction friendliness diagnose failures of predictability not merely associated with singular points, but with causally isolated regions. Plausibly, then, prediction and retrodiction unfriendliness diagnose failures of predictability that a consistent theory of Planck-scale physics should not be assumed to fix. 

Manchak and Weatherall's paradox can then be resuscitated by replacing the second and third assumptions with\footnote{\citet{schneider2024role} also modifies Manchak and Weatherall's formulation of the paradox, emphasising that the assumptions presented are about our models in semi-classical gravity. Nothing said there is incompatible with my discussion. The considerations I present urge against demanding our semi-classical models be globally hyperbolic, which would equally dissolve Schneider's formulation of the paradox.}: 

\begin{enumerate}
    \item [*2.] (*CCH) All physically reasonable spacetimes are prediction and retrodiction friendly. 
\end{enumerate}

\begin{enumerate}
    \item [*3.] (*KWL) Some physically reasonable black hole evaporation spacetime is retrodiction unfriendly.  
\end{enumerate}
(*CCH) interestingly weakens (CCH), and (*KWL) strengthens (KWL) so that an inconsistency is recovered. 

One response to this resuscitated paradox is to reject (*CCH) for the same reason we rejected (CCH): one may think that a consistent theory of Planck-scale physics will recover retrodictable evolution, even if the semi-classical limit spacetime is retrodiction unfriendly. However, some physicists take retrodiction unfriendliness to signal failures of retrodictability not easily resolved by singularity resolution. Thus, it is worthwhile to see if an information paradox can be formulated based upon (*CCH). For the remainder of this paper, I will explore whether the resuscitated paradox can be avoided by rejecting (*KWL) instead of (*CCH).

The question we must answer is then: is (*KWL) true? According to Manchak and Weatherall's definition of a black hole evaporation spacetime, all black hole evaporation spacetimes are retrodiction unfriendly. They define a black hole evaporation spacetime, roughly, as one that contains a point determined by some achronal slice in the past but causally isolated from one in the future. To see the role of this definition, consider the full statement of the KWL theorem; following \citet{lesourd2018causal}:

\begin{theorem}[Kodama-Wald-Lesourd Theorem]

Let $(M,g)$ be a non-totally imprisoning spacetime with connected, achronal, edgeless sets $\Sigma_1, \Sigma_2 \subset I^+(\Sigma_1)$ such that $\Sigma_1$ is closed and non-compact. Suppose the following: (i) $J^+(K)\cap \Sigma_2$ has compact closure, where $K = \Sigma_1 - D^-(\Sigma_2)$, (ii) there is a point $p \in D^+(\Sigma_1)- J^+(\Sigma_2)\cup J^-(\Sigma_2)$, (iii) $\Sigma_1 - V \subset J^-(\Sigma_2)$ where $V$ is a compact subset of $\Sigma_1$, (iv) $J^-(Q) \cap \Sigma_1$ has compact closure where $Q=\Sigma_2 - D^+(\Sigma_1)$. Then $(M,g)$ is not causally simple.
\end{theorem}

Causal simplicity is weaker than global hyperbolicity, so non-global hyperbolicity is proven.

The reader will have noticed that to define retrodiction friendliness I have just negated condition (ii) of the KWL theorem, and shuffled the presentation. 

Manchak and Weatherall define an evaporation spacetime as one which obeys the antecedent conditions of the KWL theorem. (KWL) then follows from the theorem. However, nothing is proven to establish (*KWL), which is assumed by definition.\footnote{As discussed in footnote \ref{fn:aproposMaudlin}, \citet{Manchak2018-MANPRA-5} are responding to a specific argument of \citet{maudlin2017information} and for this purpose the definition they use is entirely appropriate.} That is, the KWL theorem proves retrodiction unfriendliness implies a failure of global hyperbolicity, and retrodiction unfriendliness is assumed for black hole evaporation spacetimes.\footnote{\citet{lesourd2018causal} proves a stronger theorem with different antecedent conditions, but this alternative theorem also assumes an event horizon, and so does no better. An interesting open question is whether retrodiction friendliness and non-global hyperbolicity imply geodesic incompleteness. Non-global hyperbolicity alone is insufficient. Anti-de Sitter spacetime is non-globally hyperbolic and non-singular, so is a counterexample, but the niceness conditions (i) and (ii) of the definition of retrodiction friendliness mean this spacetime is not retrodiction friendly. It is plausible that the conjunction of non-global hyperbolicity with retrodiction friendliness is sufficient. If this is the case, then all non-globally hyperbolic retrodiction friendly spacetimes require singularity resolution.}

Thus, an information paradox requires positive arguments for (*CCH) and (*KWL), but no argument for (*KWL) has been given. Therefore, in the next section I consider the question: should we accept (*KWL)? Or equivalently, are our best models of black hole evaporation retrodiction unfriendly?

\subsection{Black Hole Evaporation can be Retrodiction Friendly}\label{ssec:BHEcanbeRetrodictionFriendly}

Assuming black hole evaporation spacetimes are retrodiction unfriendly amounts to the strict belief that evaporating black holes have event horizons.\footnote{\label{fn:BHdefn}The semantic question: `what is the referent of the term ``black hole''?' has many inequivalent answers \citep{Curiel2019-CURTMD-2}. All the models considered here clearly represent something minimally black hole-like.} True, the idealized spacetimes we usually use to model black holes all have event horizons: viz. Schwarzschild, Kerr, Reissner-Nordstr\"om etc. But, doubt concerning the existence of event horizons for evaporating black holes arrives on two horses. Firstly, although the apparent horizons of the `usual' stationary black hole spacetimes are null, when black hole evaporation is taken into account this apparent horizon will become timelike.\footnote{An apparent horizon is the boundary of the totally trapped region of a Cauchy surface. See \citet[p. 311]{Wald:1984rg}.} Given a timelike apparent horizon, there is a gap in the prison fence of the black hole boundary. Secondly, the singularities of black hole evaporation spacetimes should be resolved by some consistent theory of Planck-scale physics. This resolution can, potentially, allow worldlines to escape `through' the singularity. The upshot is that, although the usual suspects (Schwarzschild, Kerr, Reissner-Nordstr\"om, etc.) must be adequate for modelling the near-horizon physics of astrophysical black holes, we have reason to doubt their adequacy concerning global causal structure.\footnote{The arguments given in this section do not rule out evaporation-Schwarzschild as a model of semi-classical gravity; instead, they argue evaporation-Schwarzschild is astrophysically unreasonable, and more reasonable models are retrodiction friendly. Therefore, retrodictability is not secured for all semi-classical gravity models of black hole evaporation, and the problem remains at a conceptual level. This is analogous to the literature on closed timelike curves in general relativity, in which there are two separate sets of questions: are closed timelike curves admitted by models of general relativity? And are they admitted by astrophysically reasonable models? (E.g. \citet{Earman2009-WTHDTL, sep-time-machine, doboszewski2022rotating}). By resolving the information paradox for the second set of questions, the paradox is not necessarily resolved for the first set. I thank an anonymous reviewer for emphasising this point to me.}

Models of non-singular black holes can be traced back  at least to \citet{bardeen1968non}, with non-singular evaporating black hole models appearing at least as early as \citet{FROLOV1981}. A recent, widely discussed model is the \citet{Hayward2006} metric:

\begin{equation}\label{eq:Haywardmetric}
    ds^2 = -f(r) dt^2 + f(r)^{-1}dr^2 + r^2d\Omega^2
\end{equation}
where

\begin{equation}\label{eq:f(r)}
    f(r) = 1 - \frac{2mr^2}{r^3-2ml^2}
\end{equation}
where $l$ is a scale parameter. As $r\rightarrow \infty$ this metric approximates Schwarzschild, but deviates from it for $r \rightarrow 0$. 

The study of regular black holes is now an active research programme in physics. Although the Hayward metric, equation \eqref{eq:Haywardmetric}, is flat at $r=0$, the geodesic completeness of its maximal extension is not secure \citep{carballo2020geodesically}. However, if one uses a time-dependent mass term, and includes collapsing matter, then geodesic completeness may be recovered.\footnote{This is suggested by numerical computations of \citet{Schindler_Aguirre_Kuttner_2020}.} Regardless, alternative geodesically complete black hole models can be given (see \citet{carballo2020geodesically} and citations therein), with some well motivated by quantum gravity considerations \citep{Rovelli_2014, rovelli2024planck}. These spacetimes avoid singularities by violating the energy conditions of the singularity theorems, which is known to occur in quantum field theories, and specifically in the presence of evaporating black holes.\footnote{\citet{carballo2020geodesically} provide a taxonomy of ways of violating the energy conditions and recovering geodesically complete black hole models. Recent work \citep{wall2013generalized, bousso2023quantum} has formulated `quantum singularity theorems', which relax certain assumptions of Penrose's original work, and derive a `quantum singularity'. A full analysis of how these theorems interact with the zoology of regular black hole models is a worthwhile endeavour that goes beyond the scope of the present paper, but note that \citet{rovelli2024planck} reject the `central dogma', which is assumed in these quantum singularity theorems.} The geodesically complete evaporating black hole models discussed are globally hyperbolic by construction and as such these models are also retrodiction friendly.

This exposition is necessarily cursory; however, the models discussed should be seen as a proof of concept, rather than a concrete proposal: retrodiction friendly models of black hole evaporation in semi-classical gravity are reasonable.\footnote{An alternative class of models, motivated by `brick-wall' models \citep{hooft1985quantum}, the membrane paradigm \citep{thorne1986black}, or string-theoretic fuzzballs \citep{mathur2005fuzzball}, might modify black hole models by deleting the black hole interior. These models would then be retrodiction friendly. It is unclear whether such deleting procedures really give rise to reasonable models, especially in the case of the membrane paradigm in which there is no reason to think spacetime fails to emerge within the event horizon.} Note that, given these models are also globally hyperbolic, they could also be leveraged to undermine Manchak and Weatherall's original formulation of the information paradox.

Proofs of concept notwithstanding, given the speculative nature of any given non-singular evaporating black hole model, one may reasonably still prefer the Schwarzschild model. In this case, the best model of black hole evaporation would still be retrodiction unfriendly. However, there are examples of evaporating black hole spacetimes without event horizons, independent of any particular approach to singularity resolution; indeed they are singular spacetimes. These are deidealized models for which charge and rotation are non-zero.

The spacetime that describes a charged, rotating, stationary  black hole (i.e. a Kerr-Newman spacetime) has a timelike singularity in a black hole region bounded by an event horizon \citep[sec. 12.3]{Wald:1984rg}. Because of this timelike rather than spacelike singularity, modelling evaporation for charged and/or rotating black holes changes the causal structure significantly. \citet{kaminaga1990dynamical} and \citet{parikh1999global} study the causal structure of an evaporating charged, non-rotating black hole. They model a collapsing null shell which transitions to a charged Vaidya metric when the shell crosses the radius at which a black hole would be expected to form. Vaidya metrics are essentially black hole spacetime metrics but with a time-dependent mass term. The charged Vaidya metric is then allowed to evaporate and is finally replaced by a Minkowski metric once the mass term goes to zero. The authors find that both the singularity and the apparent horizon are timelike, and so there is no event horizon. The conformal diagram for this spacetime is given by \citet{parikh1999global} and is displayed in figure \ref{fig:EvaporationRN}. \citet{braunstein2021information} repeat this analysis for an evaporating, charged, rotating black hole and identify the same global causal structure.

These deidealized models of evaporating black holes each possess a naked singularity and thus fail to be globally hyperbolic. However, because they do not have event horizons, they violate condition (iv) of the KWL theorem. Therefore, these spacetimes are retrodiction (and prediction) friendly and so undermine (*KWL).\footnote{\citet{parikh1999global} suggest a further implication of these models: that the semi-classical spacetime can be `predicted' by the imposition of boundary conditions at the singularity. They write: ``The entire spacetime can be predicted from initial conditions if boundary conditions at the singularity are known'' (p. 1). Moreover, they suggest this would not be possible in evaporation-Schwarzschild with its spacelike singularity: ``it is not obvious how a more complete dynamical theory could replace them with something more natural'' (p. 1). Substantiating these claims is an exciting path for future research. If correct, it may further support the claim that retrodiction friendly spacetimes are more conducive to predictable evolution than retrodiction unfriendly spacetimes, thus further bolstering the arguments of section \ref{ssec: Prediction and retrodiction friendliness}. However, this remains to be seen.}

\begin{figure}
\centering
\begin{tikzpicture}[scale=2]
  
  \coordinate (O) at ( 0, 0); 
  \coordinate (S) at ( 0,-1); 
  \coordinate (N) at ( 0, 1); 
  \coordinate (T) at ( 0, 0.5); 
  \coordinate (B) at ( 0, -0.5); 
  \coordinate (E) at ( 1, 0); 

  \coordinate (X1) at ( 0.125, -0.625); 
  \coordinate (X2) at ( 0.25, -0.75); 
  \coordinate (X3) at ( 0.875, 0.125); 
  \coordinate (X4) at ( 0.25, 0.75); 

  \node[above=1,above left=0,mydarkblue,align=center] at (O)
    {$r=0$};
  \node[right=3,mydarkblue,align=center] at (1,0.02)
    {$i^0$};
  \node[below=3,,mydarkblue,align=left] at (0.02,-1)
    {$i^-$};
  \node[above=3,,mydarkblue,align=left] at (0.02,1)
    {$i^+$};
  \node[mydarkblue,above right,align=right] at (57:0.71)
    {$\If$};
  \node[mydarkblue,below right,align=right] at (-60:0.68)
    {$\Ip$};


    \draw[world line] (X1) to [out = 45, in = -45] (T);
  
  \draw[thick,mydarkblue] (T) -- (N) -- (E) -- (S) -- (B) ;

   \draw[singularity] (B) -- (T);

   \draw[particle] (X2) -- (B);

   \draw[photon2](X1) -- (X3);

   \draw[photon2](T) -- (X4);

   \fill[pattern = {crosshatch dots}] (X1) to [out = 45, in = -45] (T) to (B) to (X1) ;

\end{tikzpicture}
\caption{Conformal diagram for an evaporating charged black hole. The dotted area is bounded by the apparent horizon. The line that intersects the bottom of the singularity is a collapsing null shell. The squiggly lines are the first and last emissions of Hawking radiation.}
    \label{fig:EvaporationRN}
\end{figure}
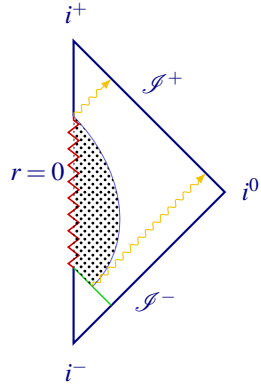

The strength of this argument against (*KWL) depends upon the accuracy of these models for representing charged and/or rotating black holes. However, the conclusions drawn should not be surprising. In the Schwarzschild case the singularity is spacelike, so there is a region of spacetime within which every future inextendible causal worldline must `intersect' the singularity and not reach $\If$. On the other hand, for a timelike singularity there is no such region. In addition, given a black hole evaporates, its apparent horizon will be timelike.  Combining the fact that causal worldlines are not forced into the timelike singularity, with the fact that the apparent horizon is timelike and shrinks, it follows that causal worldlines should be able to escape from within the apparent horizon as it evaporates past them. Therefore, in such spacetimes we should not expect there to be causally isolated regions. As such, the causal structure of the deidealized models of evaporating black holes aligns with our expectations for evaporating Kerr-Newman black holes.\footnote{Recent work \citep{brown2024evaporation} has suggested that the semi-classical approximation fails for highly charged black holes, even near the event horizon and at early times. Such suggestions deserve further investigation, so I restrict myself here to a few comments. The quantum gravity modifications discussed in the cited paper do not seem to affect the causal structure of charged, evaporating black holes, and so figure \ref{fig:EvaporationRN} may still be a good model of the relevant causal structure. Moreover, assuming astrophysical black holes are nearly neutral (given they are formed by nearly neutral collapsing stellar matter) semi-classical approximations may still be reliable. Finally, if the appropriate conclusion is that classical spacetime models of black hole evaporation are invalidated, then the paradox is still avoided by a rejection of (BHE).}


Evaporation-Schwarzschild remains retrodiction unfriendly, but it should be regarded as unphysical because \textit{all} astrophysical black holes are expected to have non-zero charge and rotation, even if these are negligible. Black holes will decharge and spin-down (e.g. \citet{Page1976}, \citet{HiscockWeems1990}) but modelling their charge and rotation as exactly zero is an idealization. The Schwarzschild metric is unstable to arbitrarily small perturbations of the charge and rotation parameters, so we should not take phenomena unique to a Schwarzschild black hole too seriously. Therefore, we do not have good reasons to believe (*KWL).  

There is a further complication: the timelike singularities of the Kerr and Reissner-Nordstr\"om metrics are expected to be unstable due to the `blue-shift heuristic' arguments of \citet{Penrose:1979}. The upshot of the arguments is that the curvature along the Cauchy horizon induced by these singularities should diverge due to an infinite blue-shift of the field modes along it. This divergence leads to the Cauchy horizon itself being singular. If these arguments are correct, then charged and rotating black holes should be expected to have null, or perhaps even spacelike, singularities, reviving worries that our models of black hole evaporation exhibit event horizons, and thus are retrodiction unfriendly. 

Three questions need to be answered to determine if instability arguments induce retrodiction unfriendliness:

\begin{enumerate}
    \item[A.] Do instability arguments apply to collapse-formed, evaporating, rotating, charged black holes?
    \item[B.] If so, does the instability generate spacelike or null singularities?
    \item[C.] If null, then are event horizons induced?
\end{enumerate}

In appendix \ref{app:instability} I sketch the blue-shift arguments and discuss the status of each of these questions in the physics literature at some length. Summarising that discussion: none of the questions has an easy or known answer. Moreover, each depends upon the behaviour of quantum fields in spacetime regions of arbitrarily large curvature near singularities, or in regions experiencing infinite blue-shift. These domains require a consistent theory of Planck-scale physics, of which we are ignorant. 

Where does this leave us? We have reasonable grounds to think that our best models of black hole evaporation will be retrodiction friendly, and thus there is no paradox. However, we cannot be confident of this conclusion: the causal structure of our best models of black hole evaporation is an outstanding issue, one for physics to decide. Significantly, whether or not our best models of black hole evaporation are retrodiction unfriendly appears to be sensitive to Planck-scale physics. This is obvious in the case of the Hayward model, but it is also the case for evaporating, charged, rotating black holes. So whether such spacetimes are retrodiction unfriendly or not turns on Planck-scale physics, of which we are ignorant. Once again, whether there is a paradox or not has been reduced to our ignorance of Planck-scale physics, and there is no paradox in ignorance. For a paradox we require a positive, well-supported argument rather than simply a large question mark, and it is hard to imagine having such an argument without knowing the physics of singularities. 

One should not be too pessimistic: we may figure out good answers to the above three questions in quantum field theory on curved spacetime. We are not close to being in this position, but it is not incomprehensible. Either way, the situation right now is one of ignorance, and so the paradox of predictability lacks the appropriate epistemic content to deliver the goods.\footnote{This position is in sympathy with the analysis of the information paradox given by \citet{Belot1999hawking}; there are interesting open questions for physics to decide, but it is doubtful that the proclamations of our best current physics are paradoxical.}

Finally, this paper has restricted its attention to formulations of the information paradox based upon predictability, expressed in terms of global spacetime structure. This does not rule out other formulations of the information paradox. For example, nothing I have said here does obvious damage to the Page-time paradox \citep{wallace_2020}, and relatedly if information is to escape the apparent horizon it will likely do so from the Planck scale, possibly implying paradoxical anti-thermodynamic behaviour in the final stages of black hole evaporation. The relationship between the analysis presented here and these other formulations of the paradox will have to wait for future work.

\section{Conclusions}

There is no information paradox due to failures of global hyperbolicity, because some non-globally hyperbolic spacetimes are physically reasonable given our expectations about Planck-scale physics. There is no paradox due to failures of retrodiction friendliness because our deidealized models of black hole evaporation are retrodiction friendly. Further work may establish that these deidealized models are also too idealized, and the instability of Cauchy horizons leads to retrodiction unfriendliness and hence a paradox. However, we are currently ignorant about whether this is the case, and it is unclear whether anything other than a consistent theory of Planck-scale physics can provide answers to these questions. We thus lack a positive argument for a paradox based upon predictability. Therefore, given the current state of our understanding of black hole evaporation, a paradox must be based upon some other feature of the evaporation process.

\begin{appendix}

\section{Proofs of Conformal Inequivalence and Mixedness}\label{app:proof}

\begin{definition}[Conformal Isometry] 
See \citet[p. 443]{Wald:1984rg}.
\end{definition}

\cfinequiv*
\begin{proof}
    For the purpose of contradiction, assume $\psi$ as above exists. Let $\gamma_1$ be a null curve in $M$ from $\Ip$ to the evaporation event, as shown in figure \ref{fig:proof_MGHD}. Let $\mathcal{A} = J^+(\gamma_1) \subset M$. 

    Suppose there exists $p \in N$ such that $p \not\in \psi(\mathcal{A})$ and $p \in J^+(\psi(\mathcal{A}))$. Then $\psi^{-1}(p) \not\in \mathcal{A}$ and $\psi^{-1}(p) \in J^+(\mathcal{A}) = \mathcal{A}$, which is a contradiction so no such $p$ exists. 

    Therefore, there are no points in the causal future of $\psi(\mathcal{A})$ not contained in $\psi(\mathcal{A})$, so this region must be bounded to the future by $\If$ and the singularity. 

    Now consider another null curve $\gamma_2 \subset \mathcal{A}$ which originates at $\Ip$ and ends at the Cauchy horizon. Define $\mathcal{A}_2 = J^+(\gamma_2)$. By the same argument $\psi(\mathcal{A}_2)$ is bounded to the future by $\If$ and the singularity. Therefore, a null curve $\gamma_3 \subset \mathcal{A}$ such that $\gamma_3 \subset J^-(\gamma_2)$ and that ends at the Cauchy horizon must be such that $\psi(\gamma_3)$ ends at the singularity. By constructing further null curves it follows $\psi(\gamma_2)$ must also end at the singularity. 
    
    However, the singularity is spacelike so for any two curves which hit it, one is not contained in the causal past of the other. Thus, $\psi(\gamma_3) \not\subset J^-(\psi(\gamma_2))$, and so $\psi$ cannot be a conformal isometry, contradicting the assumption.

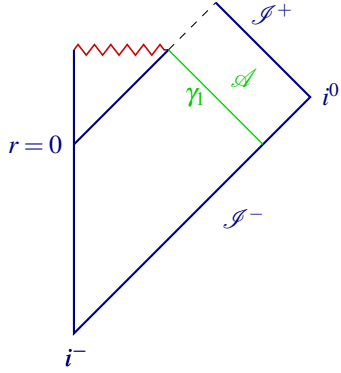
\begin{figure}
    \centering
    \begin{tikzpicture}[scale=2.5]
  
  \coordinate (O) at (0, 0); 
  \coordinate (NE)  at (0.5, 0.5); 
  \coordinate (NN) at (0.5, 1); 
  \coordinate (S)  at ( 0,-1); 
  \coordinate (N)  at ( 0, 0.5); 
  \coordinate (E)  at ( 1.25, 0.25); 
  \coordinate (X)  at ( 0.2, 0.5); 
  \coordinate (Y1) at (0, -0.3); 
  \coordinate (Y2) at (0.5, 0.6);
  \coordinate (Y3) at (0.75, 0.75);

  \draw[singularity] (N) -- node[above] {} (NE);
  
  \draw[thick,mydarkblue] (Y3) -- (E) -- (S) -- (N);
  \draw[thick,mydarkblue] (O) -- (NE);

    \path (O) -- (NE);

    \draw[dashed] (NE) -- (Y3);

  \node[above=0,left=1,mydarkblue] at (O) {$r=0$};
  \node[above=1,right=1,mydarkblue] at (E) {$i^0$};
  \node[right=1,below=1,mydarkblue] at (S) {$i^-$};
  \node[right=1,below=1,mydarkblue] at (S) {$i^-$};

  \node[mydarkblue,above right=-1] at (0.9,0.6) {$\If$};
  \node[mydarkblue,below right=-1] at (0.75,-0.3) {$\Ip$};

\tikzset{mylabel/.style  args={at #1 #2  with #3}{
    postaction={decorate,
    decoration={
      markings,
      mark= at position #1
      with  \node [#2] {#3};
 } } } }

    \draw[particle, mylabel=at 0.5 left with {$\gamma_1$}] (1, 0) -- (NE);

    \node[mygreen] at (0.9,0.35) {$\mathcal{A}$};
\end{tikzpicture}
\caption{MGHD of $\Ip$ of evaporation-Schwarzschild with $\gamma_1$ and $\mathcal{A}$ as discussed in the appendix.}\label{fig:proof_MGHD}
\end{figure}
\end{proof}

\mixed*

This theorem follows from the following lemma \citep[p. 210]{Takesaki1979}:

\begin{lemma}\label{thm:lemma}
If $\mathcal{A}$ is a C* -sub-algebra of a C*-algebra $\mathcal{B}$ and if the restriction $\omega_A$ of a
state $\omega$ on $\mathcal{B}$ to $\mathcal{A}$ is pure, then $\omega(ab)=\omega(a)\omega(b)$ for all $a \in \mathcal{A}$ and for all $b \in \mathcal{B}$ that commute with all $a\in\mathcal{A}$. 
\end{lemma}

Armed with this lemma, I now sketch of proof of Theorem \ref{thm:mixedness}:
\begin{proof}

By retrodiction unfriendliness there is a region\footnote{The definition of retrodiction unfriendliness only assumes there is a point, but a point implies a region.} of spacetime $R$ in the causal future of $\Sigma_1$ but spacelike separated from $\Sigma_2$, so by postulate 2 the respective algebras of observables commute. By retrodiction unfriendliness and postulate 1, $R$ and $\Sigma_2$ have a common causal past, so by postulate 3 some elements in their respective algebras of observables are correlated. By Lemma \ref{thm:lemma} if $\omega_{\Sigma_2}$ is pure, then no such correlation can exist, hence $\omega_{\Sigma_2}$ is mixed.
\end{proof}

\section{Instability of Timelike Singularities}\label{app:instability}

In this appendix I sketch the blue-shift arguments for Cauchy horizon instability, and then consider whether this induces retrodiction friendliness for black hole evaporation spacetimes. 

Consider the conformal diagram of a stationary Reissner-Nordstr\"om black hole, a subset of which is depicted in figure \ref{fig:RN}. Consider null rays being sent into the spacetime from $\Ip$ for advanced time $v\rightarrow\infty$. A timelike worldline in the black hole region $\BH$ propagating toward the right-hand segment of the Cauchy horizon $\mathcal{H}^C$ reaches the Cauchy horizon in finite proper time, and is extendible beyond the Cauchy horizon. However, an infinite number of the null rays sent from $\Ip$ will therefore reach the timelike worldline in finite proper time. Thus an infinite amount of energy is squeezed into finite time, leading to an infinite blue-shift along the Cauchy horizon. This infinite blue-shift should lead spacetime curvature to diverge on the Cauchy horizon and thus for it to become singular. 

\begin{figure}
    \centering
    \begin{tikzpicture}[scale=1.3]

  \coordinate (-O) at (-1, 0); 
  \coordinate (-S) at (-1,-1); 
  \coordinate (-N) at (-1, 1); 
  \coordinate (-W) at (-2, 0); 
  \coordinate (-E) at ( 0, 0); 
  \coordinate (O)  at ( 1, 0); 
  \coordinate (S)  at ( 1,-1); 
  \coordinate (N)  at ( 1, 1); 
  \coordinate (E)  at ( 2, 0); 
  \coordinate (W)  at ( 0, 0); 
  \coordinate (B)  at ( 0,-1); 
  \coordinate (T)  at ( 0, 1); 

  \coordinate (-SS) at (-1, -1.2);
  \coordinate (SS) at (1, -1.2);
  \coordinate (-NN) at (-1, 3);
  \coordinate (NN) at (1, 3);

  \draw[singularity] (-N) -- node[above,scale=0.7] {} (-NN);
  \draw[singularity] (N) -- node[above,scale=0.7] {} (NN);
  \draw[singularity] (-S) -- node[below,scale=0.7] {} (-SS);
  \draw[singularity] (S) -- node[below,scale=0.7] {} (SS);
  \draw[dashed] (-N) -- (NN) {};
  \draw[dashed] (N) -- node[pos=0.3, above right = -0.1cm, mydarkblue] {$\mathcal{H}^C$} (-NN) {};
  \draw[thick,mydarkblue] (-N) -- (-E) -- (-S) -- (-W) -- cycle;
  \draw[thick,mydarkblue] (N) -- (E) -- (S) -- (W) -- cycle;

   \node[above=1,right=1,mydarkblue] at (2,0) {$i^0$};

   \node[mydarkblue,above right=-1] at (1.5,0.5) {$\calI^+$};
   \node[mydarkblue,below right=-1] at (1.5,-0.5) {$\calI^-$};
  
  \node[inner sep=2] at (O) {I};
  \node[inner sep=2] at (0,1) {$\BH$};
  \node[inner sep=2] at (-1,0) {II};
  
\end{tikzpicture}
\caption{A subset of the conformal diagram for a stationary, charged black hole, modelled by the Reissner-Nordstr\"om metric.}\label{fig:RN}
\end{figure}
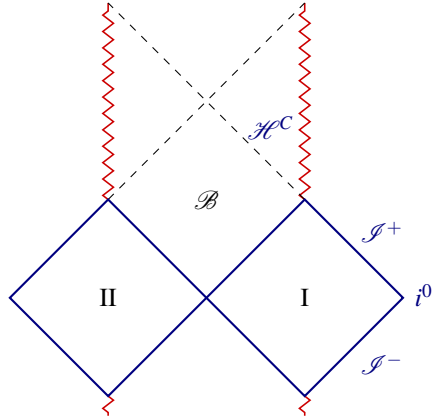

This is the core heuristic found in \citet[pp. 619-621]{Penrose:1979}. The instability of the Cauchy horizons for Reissner-Nordstr\"om and Kerr suggests the possible instability of Cauchy horizons generically. In particular, it may suggest the instability of the Cauchy horizon in the retrodiction friendly but non-globally hyperbolic models considered in section \ref{sec:retrodictability}. For instance the charged Vaidya metric, figure \ref{fig:EvaporationRN}, has a Cauchy horizon due to the timelike singularity. If perturbing this charged Vaidya metric induces a spacelike singularity, retrodiction unfriendliness will once again rear its head. 

To decide on this issue, three (difficult) questions require answers:

\begin{enumerate}
    \item[A.] Do instability arguments apply to collapse-formed, evaporating, rotating, charged black holes?
    \item[B.] If so, does the instability generate spacelike or null singularities?
    \item[C.] If null, then are event horizons induced?
\end{enumerate}

None of the answers to these questions are known. 

Concerning the first, the Cauchy horizon of the charged Vaidya metric does not necessarily induce an infinite blue-shift. To see this, note that the charged Vaidya metric does not have the causal structure of the Reissner-Nordstr\"om metric; specifically it does not have an infinite proper time being compressed into a region of finite proper time. The compression of the infinite set of signals into a ``flash'' \citep[p.621]{Penrose:1979} for an observer crossing the Cauchy horizon is the core of the blue-shift heuristics.

This is not to say that such arguments cannot be recovered in evaporating spacetimes, but whether they can is a non-trivial issue. The arguments supporting (or denying) the instability of Cauchy horizons have developed far beyond Penrose's initial heuristic (see references below), but these arguments are largely for stationary black holes. These arguments must be deidealized to the case of evaporating black holes and, as we have seen, evaporation can substantially change relevant facts about the causal structure. 

If the arguments do apply, then retrodiction unfriendliness does not immediately follow unless the instability generates spacelike singularities, which is the content of the second question. Given we do not know the answer to the first question \textit{a fortiori} we do not know whether the hypothetical horizon instabilities lead to spacelike or null singularities. But matters are even worse here because we do not even know if the induced singularities in stationary black holes are spacelike or null.

There are positive results claiming that the Cauchy horizon of a stationary charged or rotating black hole decays into a weak null singularity (\citet{Hiscock1977, PoissonIsrael1990, Ori1991, dafermos2001stability, dafermos2005interior, dafermos2014black}; for more citations see \citet{doboszewski2022rotating, chesler2019singularities}). But there are also results that suggest the development of a spacelike singularity \citep{burko1997structure, chesler2019numerical, chesler2019singularities}. Finally, all the above deals only with classical fields. For quantum fields, there is initial work in the recent literature which supports the instability, but it is insufficient to resolve the issue between null and spacelike singularities \citep{hollands2020quantum}. These are accepted as outstanding issues within the physics community. It would be a worthy project to assess whether the various analyses for stationary black holes survive deidealization to evaporating black holes. Limitations of scope prevent even a meagre attempt at such an analysis here. 

If the Cauchy horizon of an evaporating, charged, rotating black hole is unstable and if the induced singularity is spacelike, then the spacetime will be retrodiction unfriendly and a paradox of retrodictability follows. What about if the induced singularity is null? Then whether or not the spacetime is retrodiction unfriendly turns upon the form this null singularity takes, which is the content of the third question. 

We can distinguish two possible forms a null singularity could take, which are depicted in figure \ref{fig:null_singularity_forms}. One of these, on the left of the figure, has a null singularity along what would be the Cauchy horizon of the charged Vaidya metric, until this singularity meets the timelike apparent horizon and the black hole vanishes. This is a non-event-horizon-inducing null singularity. Alternatively, as depicted on the right of the figure, the null singularity could form along what would be the interior Cauchy horizon of the stationary Reissner-Nordstr\"om metric. This is an event-horizon-inducing null singularity. Only event-horizon-inducing null singularities lead to retrodiction unfriendliness.

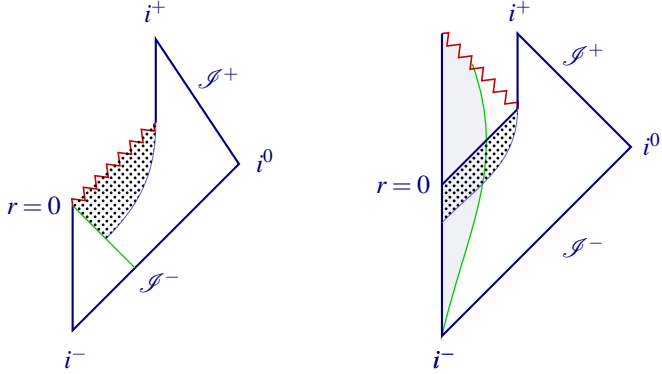
\begin{figure}

\centering

 \begin{minipage}[h]{.4\textwidth}
   \begin{tikzpicture}[scale=2.2]
  
  \coordinate (O) at ( 0, 0); 
  \coordinate (S) at ( 0,-1); 
  \coordinate (N) at ( 0.5, 0.75); 
  \coordinate (T) at ( 0.5, 0.25); 
  \coordinate (B) at ( 0, -0.25); 
  \coordinate (E) at ( 1, 0); 

  \coordinate (X1) at ( 0.2, -0.45); 
  \coordinate (X2) at ( 0.375, -0.625); 

  \node[above=0,left=1,mydarkblue] at (B) {$r=0$};
  \node[right=3,mydarkblue,align=center] at (1,0.02)
    {$i^0$};
  \node[below=3,,mydarkblue,align=left] at (0.02,-1)
    {$i^-$};
  \node[above=3,,mydarkblue,align=left] at (N)
    {$i^+$};
  \node[mydarkblue,above right,align=right] at (30:0.8)
    {$\If$};
  \node[mydarkblue,below right,align=right] at (-60:0.68)
    {$\Ip$};


    \draw[world line] (X1) to [out = 45, in = -90] (T);
  
  \draw[thick,mydarkblue] (T) -- (N) -- (E) -- (S) -- (B) ;

   \draw[singularity] (B) -- (T);

   \draw[particle] (X2) -- (B);

   \fill[pattern = {crosshatch dots}] (X1) to [out = 45, in = -90] (T) to (B) to (X1) ;

\end{tikzpicture}
\end{minipage}
\begin{minipage}[h]{.4\textwidth}
 \begin{tikzpicture}[scale=2]
  \message{evaporation-Schwarzschild}

  \coordinate (O) at (0, 0); 
  \coordinate (NE)  at (0.5, 0.5); 
  \coordinate (NN) at (0.5, 1); 
  \coordinate (S)  at ( 0,-1); 
  \coordinate (N)  at ( 0, 1); 
  \coordinate (E)  at ( 1.25, 0.25); 
  \coordinate (X)  at ( 0.2, 0.8); 
  \coordinate (X1) at (0, -0.25);
  
  \draw[particle, fill = mylightblue] (N) to (S) to[out=77,in=-70] (X);
  
  \draw[singularity] (N) -- node[above] {} (NE);
  
  \draw[thick,mydarkblue] (NE) -- (NN) -- (E) -- (S) -- (N);
  \draw[thick,mydarkblue] (O) -- (NE);

    \path (O) -- (NE); 

  \node[above=0,left=1,mydarkblue] at (O) {$r=0$};
  \node[above=1,right=1,mydarkblue] at (E) {$i^0$};
  \node[right=1,below=1,mydarkblue] at (S) {$i^-$};
  \node[right=1,above=1,mydarkblue] at (NN) {$i^+$};
  \node[right=1,below=1,mydarkblue] at (S) {$i^-$};

  \node[mydarkblue,above right=-1] at (0.75,0.75) {$\If$};
  \node[mydarkblue,below right=-1] at (0.75,-0.3) {$\Ip$};

\draw[world line] (X1) to [out = 45, in = -90] (NE);

   \fill[pattern = {crosshatch dots}] (X1) to [out = 45, in = -90] (NE) to (O) to (X1) ;

\end{tikzpicture}
    \end{minipage}
    \caption{Left: A non-event horizon-inducing null singularity. Right: An event-horizon-inducing null singularity.}\label{fig:null_singularity_forms}
\end{figure}

Are the (putative) null singularities formed by the (putative) Cauchy horizon instability of evaporating, collapse-formed black holes event-horizon-inducing or not? As with the previous questions, we do not know. 

\end{appendix}

\begin{bmhead}[Acknowledgments.]
I am grateful to John Earman, Sam Fletcher, Nick Huggett, J.B. Manchak, Bryan Roberts, David Wallace, and Jim Weatherall for feedback and discussion at various stages of this project. I am also grateful to two anonymous reviewers for their helpful comments.
\end{bmhead}

\begin{bmhead}[Declarations.]
None to declare.
\end{bmhead}

\begin{bmhead}[Funding Information.]
None to declare.
\end{bmhead}

\bibliographystyle{psalike-v1.0}
\bibliography{citations}

@article{Manchak2018-MANPRA-5,
	author = {J. B. Manchak and James Owen Weatherall},
	title = "{(Information) Paradox Regained? A Brief Comment on Maudlin on Black Hole Information Loss}",
	year = {2018},
	publisher = {Springer},
	volume = {48},
	number = {6},
	doi = {https://doi.org/10.1007/s10701-018-0170-3},
	pages = {611--627},
	journal = {Foundations of Physics}
}

@incollection{Penrose:1979,
    author = {Penrose, R.},
    title = {Singularities and Time Asymmetry},
    booktitle = "{General Relativity: An Einstein Centenary Survey}",
    editors = {Hawking, S. and Israel, W.},
    publisher = {Cambridge University Press},
    pages = {581--638},
    year = {1979}
}

@article{Manchak2011-MANWIA,
	title = {What is a Physically Reasonable Space-Time?},
	volume = {78},
	number = {3},
	pages = {410--420},
	year = {2011},
	doi = {https://doi.org/10.1086/660301},
	journal = {Philosophy of Science},
	author = {John Byron Manchak}
}

@book{Wald:1984rg,
    author = "Wald, Robert M.",
    title = "{General Relativity}",
    doi = "https://doi.org/10.7208/chicago/9780226870373.001.0001",
    publisher = "Chicago Univ. Pr.",
    address = "Chicago, USA",
    year = "1984"
}

@article{Kodama:1979vm,
    author = "Kodama, Hideo",
    title = "{Inevitability of a naked singularity associated with the black hole evaporation}",
    reportNumber = "KUNS-500",
    doi = "https://doi.org/10.1143/PTP.62.1434",
    journal = "Prog. Theor. Phys.",
    volume = "62",
    pages = "1434",
    year = "1979"
}

@incollection{wald1984black,
  title = {Black holes, singularities and predictability},
  author = {Wald, Robert M},
  booktitle = {Quantum theory of gravity. Essays in honor of the 60th birthday of Bryce S. Dewitt},
  year = {1984},
  editor = {Christensen, S.},
  publisher = {Adam Hilger Ltd.}
}

@Misc{maudlin2017information,
  author        = {Maudlin, Tim},
  title         = {{(Information) Paradox Lost}},
  year          = {2017},
  archiveprefix = {arXiv},
  copyright     = {arXiv.org perpetual, non-exclusive license},
  doi           = {https://doi.org/10.48550/arXiv.1705.03541},
  eprint        = {1705.03541},
  primaryclass  = {physics.hist-ph},
  publisher     = {arXiv},
}

@article{Belot1999hawking,
author = {Belot, G. and Earman, J. and Ruetsche, L.},
title = "{The Hawking Information Loss Paradox: The Anatomy of Controversy}",
journal = {The British Journal for the Philosophy of Science},
volume = {50},
number = {2},
pages = {189-229},
year = {1999},
doi = {https://doi.org/10.1093/bjps/50.2.189},
}

@article{Hawking:1975vcx,
    author = "Hawking, S. W.",
    title = "{Particle Creation by Black Holes}",
    doi = "https://doi.org/10.1007/BF02345020",
    journal = "Commun. Math. Phys.",
    volume = "43",
    pages = "199--220",
    year = "1975",
    note = "[Erratum: \textit{Commun. Math. Phys.} 46, 206 (1976)]"
}

@incollection{wallace_2020, place = {Cambridge}, title = {Why Black Hole Information Loss Is Paradoxical}, DOI = {https://doi.org/10.1017/9781108655705.013}, booktitle = {Beyond Spacetime: The Foundations of Quantum Gravity}, publisher = {Cambridge University Press}, author={Wallace, David}, editor = {Huggett, Nick and Matsubara, Keizo and Wüthrich, ChristianEditors}, year = {2020}, pages = {209–236}}

@article{Hawking1976,
  title = {Breakdown of predictability in gravitational collapse},
  author = {Hawking, S. W.},
  journal = {Phys. Rev. D},
  volume = {14},
  issue = {10},
  pages = {2460--2473},
  numpages = {0},
  year = {1976},
  month = {Nov},
  publisher = {American Physical Society},
  doi = {https://doi.org/10.1103/PhysRevD.14.2460}
}

@Article{kay1992principle,
  author    = {Kay, Bernard S},
  journal   = {Reviews in Mathematical Physics},
  title     = {{The principle of locality and quantum field theory on (non globally hyperbolic) curved spacetimes}},
  year      = {1992},
  number    = {spec01},
  pages     = {167--195},
  volume    = {4},
  doi       = {https://doi.org/10.1142/s0129055x92000194},
  publisher = {World Scientific},
}

@article{Yurtsever_1994,
doi = {https://doi.org/10.1088/0264-9381/11/4/016},
year = {1994},
publisher = {},
volume = {11},
number = {4},
pages = {999},
author = {Ulvi Yurtsever},
title = {Algebraic approach to quantum field theory on non-globally-hyperbolic spacetimes},
journal = {Classical and Quantum Gravity}
}

@article{Geroch1970,
    author = {Geroch, Robert},
    title = "{Domain of Dependence}",
    journal = {Journal of Mathematical Physics},
    volume = {11},
    number = {2},
    pages = {437-449},
    year = {1970},
    doi = {https://doi.org/10.1063/1.1665157},
}

@incollection{geroch1977,
  title={Prediction in General Relativity},
  author={Geroch, Robert},
  booktitle={Foundations of Space-Time Theories. Minnesota Studies in the Philosophy of Science, vol. 8},
  editor = {Earman, J. and Glymour, C. and Stachel, J.},
  pages={81--93},
  year={1977},
  publisher={University of Minnesota Press}
}

@Article{Curiel2019-CURTMD-2,
  author  = {Erik Curiel},
  journal = {Nature Astronomy},
  title   = {{The Many Definitions of a Black Hole}},
  year    = {2019},
  pages   = {27--34},
  volume  = {3},
  doi     = {https://doi.org/10.1038/s41550-018-0602-1},
}

@Article{ChoquetBruhatGeroch1969,
  author    = {Yvonne Choquet-Bruhat and Robert Geroch},
  journal   = {Communications in Mathematical Physics},
  title     = {{Global aspects of the Cauchy problem in general relativity}},
  year      = {1969},
  number    = {4},
  pages     = {329 -- 335},
  volume    = {14},
  doi       = {https://doi.org/10.1007/bf01645389},
  publisher = {Springer},
}

@book{earman1995bangs,
	title = {Bangs, Crunches, Whimpers, and Shrieks: Singularities and Acausalities in Relativistic Spacetimes},
	year = {1995},
	author = {John Earman},
	publisher = {Oxford University Press USA}
}

@article{Rovelli_2014,
	doi = {https://doi.org/10.1142/s0218271814420267},
	year = 2014,
	volume = {23},
	number = {12},
	pages = {1442026},
  
	author = {Carlo Rovelli and Francesca Vidotto},
  
	title = {Planck stars},
  
	journal = {International Journal of Modern Physics D}
}

@article{Wald1980,
  title = {Quantum gravity and time reversibility},
  author = {Wald, Robert M.},
  journal = {Phys. Rev. D},
  volume = {21},
  issue = {10},
  pages = {2742--2755},
  numpages = {0},
  year = {1980},
  month = {May},
  publisher = {American Physical Society},
  doi = {https://doi.org/10.1103/PhysRevD.21.2742},
  url = {https://link.aps.org/doi/10.1103/PhysRevD.21.2742}
}

@book{birrell_davies_1982, place={Cambridge}, series={Cambridge Monographs on Mathematical Physics}, title={Quantum Fields in Curved Space}, DOI={https://doi.org/10.1017/CBO9780511622632}, publisher={Cambridge University Press}, author={Birrell, N. D. and Davies, P. C. W.}, year={1982}, collection={Cambridge Monographs on Mathematical Physics}}

@Article{lesourd2018causal,
  author    = {Lesourd, Martin},
  journal   = {Classical and Quantum Gravity},
  title     = {Causal structure of evaporating black holes},
  year      = {2018},
  number    = {2},
  pages     = {025007},
  volume    = {36},
  doi       = {https://doi.org/10.1088/1361-6382/aaf5f8},
  publisher = {IOP Publishing},
}

@article{Schindler_Aguirre_Kuttner_2020,
  title = "{Understanding black hole evaporation using explicitly computed Penrose diagrams}",
  author = {Schindler, Joseph C. and Aguirre, Anthony and Kuttner, Amita},
  journal = {Phys. Rev. D},
  volume = {101},
  issue = {2},
  pages = {024010},
  numpages = {22},
  year = {2020},
  month = {Jan},
  publisher = {American Physical Society},
  doi = {https://doi.org/10.1103/PhysRevD.101.024010},
  url = {https://link.aps.org/doi/10.1103/PhysRevD.101.024010}
}

@Article{parikh1999global,
  author    = {Parikh, Maulik K and Wilczek, Frank},
  journal   = {Physics Letters B},
  title     = {Global structure of evaporating black holes},
  year      = {1999},
  number    = {1-2},
  pages     = {24--29},
  volume    = {449},
  doi       = {https://doi.org/10.1016/s0370-2693(99)00071-4},
  publisher = {Elsevier},
}

@Article{kaminaga1990dynamical,
  author    = {Kaminaga, Yasuhito},
  journal   = {Classical and Quantum Gravity},
  title     = {{A dynamical model of an evaporating charged black hole and quantum instability of Cauchy horizons}},
  year      = {1990},
  number    = {7},
  pages     = {1135},
  volume    = {7},
  doi       = {https://doi.org/10.1088/0264-9381/7/7/011},
  publisher = {IOP Publishing},
}

@Article{braunstein2021information,
  author    = {Braunstein, S. L. and Das, S. and Wang, ZW},
  journal   = {International Journal of Modern Physics D},
  title     = {Information recovery from evaporating black holes},
  year      = {2021},
  number    = {09},
  pages     = {2150069},
  volume    = {30},
  doi       = {https://doi.org/10.1142/s0218271821500693},
  publisher = {World Scientific},
}

@incollection{crowther2022four,
	author = {Crowther, K. and De Haro, S.},
	booktitle = {The Foundations of Spacetime Physics: Philosophical Perspectives},
	editor = {Antonio Vassallo},
	pages = {223--250},
	publisher = {Routledge},
	title = {Four Attitudes Towards Singularities in the Search for a Theory of Quantum Gravity},
	year = {2022},
    doi       = {https://doi.org/10.4324/9781003219019-12}
}

@article{HiscockWeems1990,
  title = {Evolution of charged evaporating black holes},
  author = {Hiscock, William A. and Weems, Lance D.},
  journal = {Phys. Rev. D},
  volume = {41},
  issue = {4},
  pages = {1142--1151},
  year = {1990},
  doi = {https://doi.org/10.1103/PhysRevD.41.1142}
}

@article{Page1976,
  title = {Particle emission rates from a black hole. II. Massless particles from a rotating hole},
  author = {Page, Don N.},
  journal = {Phys. Rev. D},
  volume = {14},
  issue = {12},
  pages = {3260--3273},
  year = {1976},
  doi = {https://doi.org/10.1103/PhysRevD.14.3260}
}

@Article{PoissonIsrael1990,
  author  = {Poisson, Eric and Israel, Werner},
  journal = {Phys. Rev. D},
  title   = {Internal structure of black holes},
  year    = {1990},
  pages   = {1796--1809},
  volume  = {41},
  doi     = {https://doi.org/10.1103/physrevd.41.1796},
  issue   = {6},
}

@Article{Ori1991,
  author  = {Ori, Amos},
  journal = {Phys. Rev. Lett.},
  title   = {Inner structure of a charged black hole: An exact mass-inflation solution},
  year    = {1991},
  pages   = {789--792},
  volume  = {67},
  doi     = {https://doi.org/10.1103/physrevlett.67.789},
  issue   = {7},
}

@Book{dafermos2001stability,
  author    = {Dafermos, Michael Constantine},
  publisher = {Princeton University},
  title     = {{Stability and instability of the Cauchy horizon for the spherically symmetric Einstein-Maxwell-scalar field equations}},
  year      = {2001},
  doi       = {https://doi.org/10.4007/annals.2003.158.875},
}

@Article{dafermos2005interior,
  author    = {Dafermos, Mihalis},
  journal   = {Communications on Pure and Applied Mathematics: A Journal Issued by the Courant Institute of Mathematical Sciences},
  title     = {The interior of charged black holes and the problem of uniqueness in general relativity},
  year      = {2005},
  number    = {4},
  pages     = {445--504},
  volume    = {58},
  doi       = {https://doi.org/10.1002/cpa.20071},
  publisher = {Wiley Online Library},
}

@Article{dafermos2014black,
  author    = {Dafermos, Mihalis},
  journal   = {Communications in Mathematical Physics},
  title     = {Black holes without spacelike singularities},
  year      = {2014},
  pages     = {729--757},
  volume    = {332},
  doi       = {https://doi.org/10.1007/s00220-014-2063-4},
  publisher = {Springer},
}

@InCollection{doboszewski2022rotating,
  author    = {Doboszewski, Juliusz},
  booktitle = {The foundations of spacetime physics},
  publisher = {Routledge},
  title     = {Rotating black holes as time machines: An interim report},
  year      = {2022},
  pages     = {133--152},
  doi       = {https://doi.org/10.4324/9781003219019-7},
}

@Article{chesler2019numerical,
  author  = {Chesler, Paul M and Narayan, Ramesh and Curiel, Erik},
  journal = {Physical Review D},
  title   = {{Numerical evolution of shocks in the interior of Kerr black holes}},
  year    = {2019},
  number  = {8},
  pages   = {084033},
  volume  = {99},
  doi     = {https://doi.org/10.1103/physrevd.99.084033},
}

@Article{chesler2019singularities,
  author  = {Chesler, Paul M and Narayan, Ramesh and Curiel, Erik},
  journal = {Classical and Quantum Gravity},
  title   = {{Singularities in Reissner--Nordstr{\"o}m black holes}},
  year    = {2019},
  number  = {2},
  pages   = {025009},
  volume  = {37},
  doi     = {https://doi.org/10.1088/1361-6382/ab5b69},
}

@Article{burko1997structure,
  author  = {Burko, Lior M},
  journal = {Physical review letters},
  title   = {{Structure of the black hole's Cauchy-horizon singularity}},
  year    = {1997},
  number  = {25},
  pages   = {4958},
  volume  = {79},
  doi     = {https://doi.org/10.1103/physrevlett.79.4958},
}

@Article{hollands2020quantum,
  author  = {Hollands, Stefan and Wald, Robert M and Zahn, Jochen},
  journal = {Classical and Quantum Gravity},
  title   = {{Quantum instability of the Cauchy horizon in Reissner--Nordstr{\"o}m--deSitter spacetime}},
  year    = {2020},
  number  = {11},
  pages   = {115009},
  volume  = {37},
  doi     = {https://doi.org/10.1088/1361-6382/ab8052},
}

@Article{janssen2022quantum,
  author    = {Janssen, Daan W},
  journal   = {Communications in Mathematical Physics},
  title     = {Quantum fields on semi-globally hyperbolic space--times},
  year      = {2022},
  number    = {2},
  pages     = {669--705},
  volume    = {391},
  doi       = {https://doi.org/10.1007/s00220-022-04328-7},
  publisher = {Springer},
}

@article{Thebault_2023,
doi = {https://doi.org/10.1088/1361-6382/acb752},
url = {https://dx.doi.org/10.1088/1361-6382/acb752},
year = {2023},
volume = {40},
number = {5},
pages = {055007},
author = {Karim P Y Th\'ebault},
title = {Big bang singularity resolution in quantum cosmology},
journal = {Classical and Quantum Gravity}
}

@InProceedings{bojowald2007singularities,
  author    = {Bojowald, Martin},
  booktitle = {AIP Conference Proceedings},
  title     = {Singularities and quantum gravity},
  year      = {2007},
  number    = {1},
  pages     = {294--333},
  volume    = {910},
  doi       = {https://doi.org/10.1063/1.2752483},
}

@Article{Hiscock1977,
  author   = {Hiscock, William A.},
  journal  = {Phys. Rev. D},
  title    = {Stress-energy tensor near a charged, rotating, evaporating black hole},
  year     = {1977},
  pages    = {3054--3057},
  volume   = {15},
  doi      = {https://doi.org/10.1103/physrevd.15.3054},
  issue    = {10},
  numpages = {0},
}

@Article{Hayward2006,
  author   = {Hayward, Sean A.},
  journal  = {Phys. Rev. Lett.},
  title    = {Formation and Evaporation of Nonsingular Black Holes},
  year     = {2006},
  pages    = {031103},
  volume   = {96},
  doi      = {https://doi.org/10.1103/physrevlett.96.031103},
  issue    = {3},
  numpages = {4},
}

@inproceedings{bardeen1968non,
  title={Non-singular general relativistic gravitational collapse},
  author={Bardeen, James},
  booktitle={Proceedings of the 5th International Conference on Gravitation and the Theory of Relativity},
  pages={87},
  year={1968}
}

@Article{FROLOV1981,
  author  = {V.P. Frolov and G.A. Vilkovisky},
  journal = {Physics Letters B},
  title   = {Spherically symmetric collapse in quantum gravity},
  year    = {1981},
  number  = {4},
  pages   = {307-313},
  volume  = {106},
  doi     = {https://doi.org/10.1007/978-1-4613-2701-1\_19},
}

@Book{earman1986primer,
  author    = {Earman, John},
  publisher = {Springer Science \& Business Media},
  title     = {A primer on determinism},
  year      = {1986},
  doi       = {https://doi.org/10.1007/978-94-010-9072-8},
}

@Article{schneider2024role,
  author    = {Schneider, Mike D},
  journal   = {Philosophy of Science},
  title     = {A Role for the Fauxrizon in the Semiclassical Limit of a Fuzzball},
  year      = {2024},
  number    = {1},
  pages     = {225--242},
  volume    = {91},
  doi       = {https://doi.org/10.1017/psa.2023.93},
  publisher = {Cambridge University Press},
}

@Article{manchak2008prediction,
  author    = {Manchak, John Byron},
  journal   = {Foundations of physics},
  title     = {Is prediction possible in general relativity?},
  year      = {2008},
  pages     = {317--321},
  volume    = {38},
  doi       = {https://doi.org/10.1007/s10701-008-9204-6},
  publisher = {Springer},
}

@Book{Takesaki1979,
  author    = {Takesaki, M.},
  publisher = {Springer New York},
  title     = {Theory of Operator Algebras 1},
  year      = {1979},
  doi       = {https://doi.org/10.1007/978-3-662-10451-4},
}

@Article{carballo2020geodesically,
  author  = {Carballo-Rubio, Ra{\'u}l and Di Filippo, Francesco and Liberati, Stefano and Visser, Matt},
  journal = {Physical Review D},
  title   = {Geodesically complete black holes},
  year    = {2020},
  number  = {8},
  pages   = {084047},
  volume  = {101},
  doi     = {https://doi.org/10.1103/physrevd.101.084047},
}

@article{rovelli2024planck,
  title={{Planck stars, White Holes, Remnants and Planck-mass quasi-particles. The quantum gravity phase in black holes' evolution and its manifestations}},
  author={Rovelli, Carlo and Vidotto, Francesca},
  journal={arXiv preprint arXiv:2407.09584},
  year={2024}
}

@Article{wall2013generalized,
  author  = {Wall, Aron C},
  journal = {Classical and Quantum Gravity},
  title   = {The generalized second law implies a quantum singularity theorem},
  year    = {2013},
  number  = {16},
  pages   = {165003},
  volume  = {30},
  doi     = {https://doi.org/10.1088/0264-9381/30/16/165003},
}

@Article{bousso2023quantum,
  author  = {Bousso, Raphael and Shahbazi-Moghaddam, Arvin},
  journal = {Physical Review D},
  title   = {Quantum singularities},
  year    = {2023},
  number  = {6},
  pages   = {066002},
  volume  = {107},
  doi     = {https://doi.org/10.1103/physrevd.107.066002},
}

@Article{hooft1985quantum,
  author    = {'t Hooft, Gerard},
  journal   = {Nuclear Physics B},
  title     = {On the quantum structure of a black hole},
  year      = {1985},
  pages     = {727--745},
  volume    = {256},
  doi       = {https://doi.org/10.1016/0550-3213(85)90418-3},
  publisher = {Elsevier},
}

@book{thorne1986black,
  title={Black Holes: The Membrane Paradigm},
  author={Thorne, Kip S and Price, Richard H and MacDonald, Douglas A},
  year={1986},
  publisher={Yale University Press}
}

@Article{mathur2005fuzzball,
  author    = {Mathur, Samir D},
  journal   = {Fortschritte der Physik: Progress of Physics},
  title     = {The Fuzzball proposal for black holes: An Elementary review},
  year      = {2005},
  number    = {7-8},
  pages     = {793--827},
  volume    = {53},
  doi       = {https://doi.org/10.1002/prop.200410203},
  publisher = {Wiley Online Library},
}

@Article{brown2024evaporation,
  author    = {Brown, Adam R. and Iliesiu, Luca V. and Penington, Geoff and Usatyuk, Mykhaylo},
  journal   = {Journal of High Energy Physics},
  title     = {The evaporation of charged black holes},
  year      = {2026},
  issn      = {1029-8479},
  month     = Jan,
  number    = {1},
  volume    = {2026},
  doi       = {https://doi.org/10.1007/jhep01(2026)109},
  publisher = {Springer Science and Business Media LLC},
}

@article{Earman2009-WTHDTL,
	author = {John Earman and Chris Smeenk and Christian W\"{u}thrich},
	doi = {https://doi.org/10.2307/40271293},
	journal = {Synthese},
	number = {1},
	pages = {91--124},
	publisher = {Springer},
	title = {Do the Laws of Physics Forbid the Operation of Time Machines?},
	volume = {169},
	year = {2009}
}

@InCollection{sep-time-machine,
	author       =	{Earman, John and Wüthrich, Christian and Manchak, JB},
	title        =	{{Time Machines}},
	booktitle    =	{The {Stanford} Encyclopedia of Philosophy},
	editor       =	{Edward N. Zalta and Uri Nodelman},
	howpublished =	{\url{https://plato.stanford.edu/archives/sum2024/entries/time-machine/}},
	year         =	{2024},
	edition      =	{{S}ummer 2024},
	publisher    =	{Metaphysics Research Lab, Stanford University}
}

@article{WeatherallForthcoming-WEAWDG,
	author = {James Owen Weatherall},
	doi = {https://doi.org/10.1017/psa.2022.98},
	journal = {Philosophy of Science},
	pages = {1--10},
	title = {Where Does General Relativity Break Down?},
	year = {forthcoming}
}

\end{document}